%% file: ICALP2026.tex
\documentclass[a4paper,UKenglish,cleveref, autoref, thm-restate]{lipics-v2021}
\nolinenumbers

\title{Unambiguisability and Register Minimisation of Min-Plus Models} 

\author{Shaull Almagor}{Department of Computer Science, Technion, Israel \and \url{https://shaull.cswp.cs.technion.ac.il/}} {shaull@technion.ac.il}{https://orcid.org/0000-0001-9021-1175}{}

\author{Guy Arbel}{Department of Computer Science, Technion, Israel} {guy.arbel@campus.technion.ac.il}{}{}

\author{Sarai Sheinvald}{Department of Computer Science, Technion, Israel} {surke@technion.ac.il}{https://orcid.org/0000-0002-0524-7390}{}

\authorrunning{S. Almagor,  G. Arbel and S. Sheinvald} 

\Copyright{Shaull Almagor, Guy Arbel and Sarai Sheinvald} 

\ccsdesc[100]{\textcolor{red}{Replace ccsdesc macro with valid one}} 

\keywords{Automata, Weighted Automata, Determinisation, Unambiguous, Unambiguisation, Tropical, Min Plus} 

\category{} 

\relatedversion{} 

\acknowledgements{This research was supported by the ISRAEL SCIENCE FOUNDATION (grant No. 989/22)}

\EventEditors{Sayan Bhattacharya, Danupon Nanongkai, Michael Benedikt, and Gabriele Puppis}
\EventNoEds{4}
\EventLongTitle{53rd International Colloquium on Automata, Languages, and Programming (ICALP 2026)}
\EventShortTitle{ICALP 2026}
\EventAcronym{ICALP}
\EventYear{2026}
\EventDate{July 7--10, 2026}
\EventLocation{Royal Holloway, University of London, Egham, United Kingdom}
\EventLogo{}
\SeriesVolume{374}
\ArticleNo{164}

\usepackage[dvipsnames]{xcolor}
\usepackage{thmtools}
\usepackage{graphicx} 
\usepackage{amsmath}
\usepackage{stmaryrd}
\usepackage{tikz}
\usetikzlibrary{decorations.pathmorphing, backgrounds,shapes.geometric, arrows, positioning, calc}
\usetikzlibrary{decorations.pathreplacing}
\usetikzlibrary{automata, positioning}
\usepackage{amssymb}
\usepackage{etoolbox}
\usepackage{xspace}
\usepackage{upgreek}
\usepackage[colorinlistoftodos]{todonotes}
\usepackage{subcaption}
\usepackage{circledsteps}
\usepackage{cancel}
\tikzstyle{block} = [rectangle, draw, fill=blue!20, 
    text centered, rounded corners, minimum height=1cm]
\tikzstyle{line} = [draw, -latex']

\input{macros}

\newcommand{\unable}{unambiguisable\xspace}
\newcommand{\unability}{unambiguisability\xspace}

\theoremstyle{definition}
\newtheorem{problem}{Problem}

\begin{document}

\maketitle

\begin{abstract}
We study the unambiguisability problem for min-plus (tropical) weighted automata (WFAs), and the register-minimisation problem for tropical Cost Register Automata (CRAs), which are expressively-equivalent to WFAs. Both problems ask whether the “amount of nondeterminism’’ in the model can be reduced. 
We show that WFA unambiguisability is decidable for tropical WFAs. Our proof is via reduction to WFA determinisability, which was recently shown to be decidable. To obtain this reduction, we develop a characterisation of unambiguisability via gaps between runs.
On the negative side, we show that CRA register minimisation is undecidable already for inputs with 7 registers, and hence also for any larger fixed number of registers.
\end{abstract}

\newpage

\section{Introduction}
\label{sec:intro}
Weighted Finite Automata (WFAs) are a popular quantitative computational model, defining functions from words to values~\cite{droste2009handbook,schutzenberger1961definition,chatterjee2010quantitative,Almagor2020Whatsdecidableweighted}. 
The semantics of WFAs are typically defined over a \emph{semiring}, with the most prominent settings being the \emph{field of rationals} $(\bbQ,+,\times)$, and the \emph{tropical semiring} $(\bbZ\cup \{\infty\},\min,+)$. 
Both semirings yield WFAs that are useful for modelling certain aspects of systems, with a wide spectrum of applications (see \cite{daviaud2020register,droste2009handbook,Almagor2020Whatsdecidableweighted} and references therein).
For example, the rational field can be used to define \emph{probabilistic automata}~\cite{paz2014introduction}, whereas the tropical semiring allows reasoning about the optimal way of using resources (e.g., energy consumption), since the semantics is to take the minimal weighted run among all the runs on a word.
Famously, tropical WFAs have been key to resolving the star-height conjecture~\cite{kirsten2005distance,Has82,Has00,LP04}.

As with many computational models, reasoning about WFAs becomes harder in the presence of nondeterminism. For example, equivalence of tropical WFAs is undecidable for nondeterministic automata, but decidable for deterministic ones~\cite{Kro94,Almagor2020Whatsdecidableweighted}.
Unlike Boolean automata, nondeterministic WFAs are strictly more expressive than their deterministic fragment for most semirings. 
Accordingly, a natural problem for WFAs is the \emph{determinisability problem}: given a WFA $\cA$, is there a deterministic WFA $\cD$ such that ${\cA}\equiv {\cD}$?
This problem has a rich history dating back to the 1990s~\cite{mohri1994compact,mohri1997finite} (see~\cite{almagor2026determinization} for more details). It was recently shown to be decidable for the rational field~\cite{bell2023computing,jecker2024determinisation}, and even more recently for the tropical semiring~\cite{almagor2026determinization}. 

The ``amount'' of nondeterminism can be measured in various ways. The most prominent is \emph{ambiguity}: a WFA is \emph{unambiguous} if every word has at most one accepting run. We can similarly define $k$-ambiguous, finitely ambiguous and polynomially ambiguous~\cite{book1971ambiguity}.
Unambiguous WFAs are strictly more expressive than deterministic WFAs, but retain some nice closure and algorithmic properties~\cite{mohri1997finite}. As such, a natural question is \emph{\unability}\footnote{Equally fun mouthfuls include: ``disambiguisability'', ``unambiguousability'', etc.}: given a WFA, is there an equivalent unambiguous WFA? 
For polynomially-ambiguous tropical automata, this problem was shown to be decidable in~\cite{kirsten2009deciding}. In addition, it is decidable over the rational field~\cite{bell2023computing}. 

In this work, we resolve the decidability of this problem for tropical WFAs with unrestricted nondeterminism, by reducing it to the determinisability problem. 
Note that for most models, determinisability and unambiguisability have been resolved in tandem, by first deciding unambiguisability, and then deciding determinisability on the equivalent unambiguous, if it exists. Interestingly, for tropical WFAs determinisation is resolved directly, while unambiguisability remained open. It is somewhat surprising, therefore, that its solution is via reduction \emph{to} determinisability.

Another measure of nondeterminism in WFAs stems from a closely related model -- that of Cost Register Automata (CRAs) with linear register updates~\cite{alur2013regular}. A cost register automaton has a deterministic control, equipped with several registers that hold values. At each step, the registers' contents is manipulated according to some semiring actions (in our case, $(\min,+)$). 
Over semirings, CRAs and WFAs are equally expressive~\cite{alur2013regular}.
A natural decision problem about CRAs is \emph{register minimisation}: given a CRA with $k$ registers, is there an equivalent CRA with $k'<k$ registers? 
We refine the results of~\cite{alur2013regular} and show that the number of registers in a CRA corresponds to the \emph{width} of a WFA: the maximal number of states that can be reached simultaneously (nondeterministically). This measure is incomparable with ambiguity as a measure for nondeterminism, but retains the flavour that width $1$ is exactly determinism. 
Using this equivalence, we show that register minimisation for CRAs is undecidable already for inputs with $7$ registers, and the construction extends to any larger fixed number of registers. 
We remark that for CRAs over the rational field, this problem is decidable~\cite{benalioua2024minimizing}.

 \paragraph*{Paper Organisation and Contributions}
In \cref{sec:prelim} we lay down basic definitions and recall some results about unambiguous WFAs. 
In \cref{sec:unambiguisability equiv characterization} we introduce a novel characterisation of unambiguisable WFAs via a notion of ``gaps''. 
In \cref{sec:unabmig and determin} we present our first main contribution -- a reduction from WFA \unability to WFA determinisability. In particular, this shows the decidability of \unability based on the recent breakthrough~\cite{almagor2026determinization}.
In \cref{sec: CRA register minimization} we show that register minimisation for CRAs is undecidable. 
We conclude with a discussion in \cref{sec:discussion}. 
Detailed proofs appear in the appendix.

\section{Preliminaries}
\label{sec:prelim}
For $k\in \bbN$ denote $[k]=\{1,\ldots,k\}$. For an alphabet $\Sigma$, we denote by $\Sigma^*$ (resp. $\Sigma^+$) the set of finite words (resp. non-empty finite words) over $\Sigma$. 
For a word $w=\sigma_1\cdots\sigma_n\in \Sigma^*$, we denote its length by $|w|=n$ and the set of its prefixes by $\pref(w)$. We write $w[i,j]=\sigma_i\cdots \sigma_j$ for the infix of $w$ corresponding to $1\le i\le j\le |w|$.
For a letter $\sigma\in \Sigma$ we denote by $\numin{\sigma}(w)$ the number of occurrences of $\sigma$ in $w$.

We denote by  $\bbNinf$ and $\bbZinf$ the sets $\bbN\cup \{\infty\}$ and $\bbZ\cup \{\infty\}$, respectively. We extend the addition and $\min$ operations to $\infty$ in the natural way: $a+\infty=\infty$ and $\min\{a,\infty\}=a$ for all $a\in \bbZinf$. By $\arg\min\{f(x)\mid x\in A\}$ we mean the set of elements in $A$ for which $f(x)$ is minimal for some function $f$ and set $A$.

 \paragraph*{Weighted Automata}
A \emph{$(\min,+)$ Weighted Finite Automaton} (WFA for short) is a tuple $\cA= \tup{Q,\Sigma, q_0, \Delta, F}$ with the following components:
\begin{itemize}
    \item $Q$ is a finite set of \emph{states}.
    \item $\Sigma$ is a finite \emph{alphabet}.
    \item $q_0\in Q$ is the \emph{initial} state (in \cref{sec: CRA register minimization} we allow a set $Q_0$ of initial states).\footnote{Having a set of initial states does not add expressiveness, as it can be replaced by a single initial state that simulates the first transition from the entire set.}
    \item $\Delta\subseteq Q\times \Sigma\times \bbZinf\times Q$ is a transition relation such that for every $p,q\in Q$ and $\sigma\in \Sigma$ there exists exactly\footnote{This is without loss of generality: if there are two transitions with different weights, the higher weight can always be ignored in the $(\min,+)$ semantics. Missing transitions can be introduced with weight $\infty$.} one weight $c\in \bbZinf$ such that $(p,\sigma,c,q)\in \Delta$.
    \item $F\subseteq Q$ is a set of \emph{accepting states}.
\end{itemize}
If for every $p\in Q$ and $\sigma\in \Sigma$ there exists at most one transition $(p,\sigma,c,q)$ with $c\neq \infty$, then $\cA$ is called \emph{deterministic}. 
We denote by $\norm{\cA}$ the maximal absolute value of any weight in $\Delta$.


 \paragraph*{Runs}
A \emph{run} of $\cA$ is a sequence of transitions $\rho=t_1,t_2,\ldots,t_m$ where $t_i=(p_i,\sigma_i,c_i,q_i)$ such that $q_i=p_{i+1}$ for all $1\le i<m$ and $c_i <\infty$ for all $1\le i \le m$.
We say that $\rho$ is a run \emph{on the word $w=\sigma_1\cdots\sigma_m$ from $p_1$ to $q_m$}, and we denote $\rho:p_1\runsto{w}q_m$. 
For an infix $x=w[i,j]$ we denote the corresponding infix of $\rho$ by $\rho[i,j]=t_i,\ldots,t_j$ (and sometimes by $\rho[x]$, if this clarifies the indices).
The \emph{weight} of the run $\rho$ is $\weight(\rho)=\sum_{i=1}^m c_i$. A run from the initial state is \emph{accepting} if it ends in $F$, i.e., if $p_1=q_0$ and $q_m\in F$. When the start state is explicit, we say that the run is accepting from that start state if it ends in $F$.

For a word $w$, we abuse the name of the WFA as the function it describes, and denote by ${\cA}(w)$ the weight assigned by $\cA$ to $w$, which is the minimal weight of an accepting run of $\cA$ on $w$. For convenience, we introduce some auxiliary notations.
For a word $w\in \Sigma^*$ and sets of states $Q_1,Q_2\subseteq Q$, denote
\[\minweight_\cA(Q_1\runsto{w} Q_2)=\min\{\weight(\rho)\mid \exists q_1\in Q_1,q_2\in Q_2,\ \rho:q_1\runsto{w}q_2\}\]
If $Q_1$ or $Q_2$ are singletons, we denote them by a single state (e.g., $\minweight_\cA(P\runsto{w} q)$ for some set $P\subseteq Q$ and state $q$).
Then, we can define
$\cA(w)=\minweight_\cA(q_0\runsto{w} F)$.
If there are no accepting runs on $w$, then $\cA(w)=\infty$.
The function $\cA:\Sigma^*\to \bbZinf$ can be seen as the weighted analogue of the \emph{language} of an automaton.

\begin{remark}[On initial and final weights and states]
    \label{rmk: initial and final weights}
    Weighted automata are often defined with initial and final weights, i.e., 
    $q_0$ is replaced with an initial vector $\init\in \bbZinf^Q$ (and in particular may have several initial states with finite weight), and there are designated accepting states or a final weight vector $\fin\in \bbZinf^Q$.
    Then, the weight of a run also includes the initial weight and final weight (which may be $\infty$).

    In \cref{sec:no need for init and fin and acc} we show that the \unability problem for this general model can be reduced to that of our setting. Therefore, it is sufficient to consider our model, without initial and final weights, and with a single initial state and a single accepting state (we use the latter assumption only in \cref{sec:unabmig and determin}). 
\end{remark}

We write $p\runsto{w}q$ when there exists some run $\rho$ 
such that $\rho:p\runsto{w}q$.
We lift this notation to concatenations of runs, e.g., $\rho: p\runsto{x}q\runsto{y}r$ means that $\rho$ is a run on $xy$ from $p$ to $r$ that reaches $q$ after the prefix $x$. We also incorporate this to $\minweight$ by writing e.g., $\minweight(q\runsto{x}p\runsto{y}r)$ to mean the minimal weight of a run $\rho:q\runsto{x}p\runsto{y}r$.

A WFA is \emph{trim} if every state is reachable from $q_0$ by some run and can reach an accepting state with some run. Note that states that do not satisfy this can be found in polynomial time (by simple graph search), and can be removed from the WFA without changing the weight of any accepted word. Throughout this paper, we assume that all WFAs are trim.

 \paragraph*{Determinisability}
We say that WFAs $\cA$ and $\cB$ are \emph{equivalent} if $\cA(w)= \cB(w)$ for every word $w$.
A WFA $\cA$ is \emph{determinisable} if it is equivalent to some deterministic WFA. Our first point of comparison is the following problem.
\begin{problem}[WFA Determinisation]
\label{prob:determinisation}
    Given a WFA $\cA$, decide whether $\cA$ is determinisable.
\end{problem}
This problem was recently shown to be decidable~\cite{almagor2026determinization}. 

 \paragraph*{Unambiguisability}
A WFA $\cA$ is \emph{unambiguous} if every word has at most one accepting run. Otherwise it is \emph{ambiguous}.
We say that $\cA$ is \emph{\unable} if it is equivalent to some unambiguous WFA. Our central object of study is the following problem.

\begin{problem}[WFA Unambiguisability]
\label{prob:unambiguisability}
    Given a WFA $\cA$, decide whether $\cA$ is \unable.
\end{problem}

\section{A Characterisation of Unambiguisability}
\label{sec:unambiguisability equiv characterization}
It is well-known that determinisability of WFA can be characterised by means of \emph{gaps}~\cite{filiot2017delay,almagor2026determinization}, namely by how far two potentially-minimal runs can get away from one another. We defer the discussion about this type of gaps to \cref{sec:determinisability gap characterisation}.

We now present an analogous characterisation for \unability. To distinguish the terms, we dub this characterisation \emph{\Utype gaps} (where $\mathfrak{U}$ stands for ``$\mathfrak{U}$nambiguous'').
Intuitively, we show that $\cA$ is \unable if and only if there is some bound $B\in \bbN$ such that any two accepting runs on a word $xy$ are no farther than $B$ apart after reading $x$, if the higher run can still become minimal after reading $y$. Conversely, $\cA$ is not \unable iff for every $B$ there exists a word $xy$ on which there exist two accepting runs that are farther than $B$ apart after reading $x$, and such that the higher run becomes minimal. We call $xy$ a \emph{witness to unambiguisability}. Such a witness is depicted in \cref{fig:utype gap witness}.

\begin{definition}[\Utype $B$-Gap Witness]
\label{def:U type B gap witness}
    Consider a WFA $\cA=\tup{Q,\Sigma,q_0,\Delta,F}$. For $B\in \bbN$, a \emph{\Utype $B$-gap witness over an alphabet $\Sigma$} consists of a pair of words $x,y\in \Sigma^*$ and states $p_1,q_1\in Q$, $p_2,q_2\in F$ such that there exist runs $\rho:q_0\runsto{x}p_1\runsto{y}p_2$ and $\chi:q_0\runsto{x}q_1\runsto{y}q_2$ and the following holds.
    \begin{itemize} 
        \item $\minweight(q_0\runsto{x}Q)=\weight(\chi[x])$, i.e., the prefix $\chi[x]:q_0\runsto{x}q_1$ is a minimal-weight run on $x$.
        \item $\minweight(q_0\runsto{xy}F)=\weight(\rho)$, i.e., $\rho$ is a minimal accepting\footnote{Note that there may be lower non-accepting runs.} run on $xy$.
        \item $\weight(\rho[x])-\weight(\chi[x])> B$, i.e., after reading $x$, the run $\rho$ is at least $B$ above the minimal run $\chi$.
    \end{itemize}
\end{definition}
We say that a WFA $\cA$ has \emph{\Utype gaps bounded by $B$} if there are no \Utype gap witnesses whose gap is greater than $B$; equivalently, there are no \Utype $(B+1)$-gap witnesses.
For brevity, we refer to \Utype gaps simply as ``gaps'' throughout this section. 
In \cref{sec:unabmig and determin}, we restore the \Utype notation as it is needed there.

The characterisation is as follows.
 \begin{restatable}{theorem}{unambigchar}
    \label{thm:unambig iff bounded U gap}
     Consider a WFA $\cA$, then $\cA$ is \unable if and only if there exists $B\in \bbN$ such that $\cA$ has gaps bounded by $B$.
\end{restatable}
The detailed proof is given in \cref{apx:thm unambig iff bounded U gap}. We present the intuition here.
\subsection{$\cA$ is Unambiguisable $\implies$ Bounded Gaps}
\label{sec:unambig implies bounded gap}
Let $\cA= \tup{Q,\Sigma, q_0, \Delta, F}$ and assume $\cA$ is \unable. Let $\cU=\tup{S,\Sigma,s_0,\Lambda,G}$ be an equivalent unambiguous WFA. 
In \cref{apx:basic properties of unambig} we recall standard results about unambiguous WFAs, namely that they can be negated. Additionally, WFAs can be summed using a standard product construction. We can therefore obtain a WFA $\cB$ for ``$\cA-\cU$'', so that 
for every $w\in \Sigma^*$ we have that either $\cA(w)=\cU(w)=\cB(w)=\infty$, or $\cB(w)=\cA(w)-\cU(w)=0$. Thus every finite value of $\cB$ is $0$.

Let $\bigM=\max\{\norm{\cA},\norm{\cU},\norm{\cB}\}$ denote the maximal weight appearing in any of $\cA,\cU$ and $\cB$ in absolute value. Thus, in a single transition, any run of these WFAs can change the weight by at most $\bigM$.

Assume by way of contradiction that $\cA$ does not have bounded gaps. In particular, there exists a $B$-gap witness for $B>2|S||Q|\bigM+1$, given by $x,y\in \Sigma^*$, $p_1,q_1\in Q,p_2,q_2\in F$ and runs $\rho:q_0\runsto{x}p_1\runsto{y}p_2$ and $\chi:q_0\runsto{x}q_1\runsto{y}q_2$ as per \cref{def:U type B gap witness}. Since $\cU$ is unambiguous, we can think of $\rho$ and $\chi$ as runs of $\cB$ with the same gap (i.e., the sum with $\cU$ does not affect the gap). Thus, the runs both end with weight $0$.

This gap size therefore implies that either $\chi[x]$ becomes very negative, or $\rho[x]$ becomes very positive (see \cref{fig:unambiguisable to bounded gaps}). In the former case, $\chi[x]$ is so low that we can find a short accepting suffix that leads to a negative-weight word, which is a contradiction.
In the latter case, reading $y$ after $\rho[x]$ must take a negative cycle, which we can pump to obtain a negative-weight word, which is again a contradiction.
We conclude that $\cB$ has bounded gaps.



\begin{figure}[ht]
\centering
\begin{subfigure}[b]{0.45\textwidth}
\centering
\begin{tikzpicture}[
    >=stealth,
    thick,
    every node/.style={inner sep=1pt},
    dot/.style={circle,draw,inner sep=1pt},
    mindot/.style={circle,draw,double,inner sep=1pt}
]

  \node[dot]   (q0) at (0,0)   {$q_0$};
  \node[dot]   (q1) at (3.5,-1.4)   {$q_1$};
  \node[dot]   (p1) at (3.5,0.6) {$p_1$};
  \node[mindot] (p2) at (6,0)   {$p_2$};
  \node[mindot] (q2) at (6,0.7)   {$q_2$};
  \node[mindot] (q2b) at (5,-0.8)   {$q'_2$};

  \draw[->,dashed] (q0) -- node[above,pos=0.7] {$\chi$} (q1);
  \draw[->] (q0) -- node[above,pos=0.7] {$\rho$} (p1);
  \draw[dotted] (q0) -- node[above,pos=0.7] {$\chi$} (p2);

 \draw[<->,blue] (q1) -- node[fill=white,
    inner sep=2pt
] {$>\!B$} (p1);

  \draw[->] (p1) -- (p2);
  \draw[dashed,->] (q1) -- (q2);
  \draw[dashed,->] (q1) -- node[below,pos=0.7] {$y'$} (q2b);

  \draw[decorate,decoration={brace,mirror,amplitude=4pt}]
        (0,-1.8) -- node[below=4pt] {$x$} (3.5,-1.8);
  \draw[decorate,decoration={brace,mirror,amplitude=4pt}]
        (3.5,-1.8) -- node[below=4pt] {$y$} (6,-1.8);

  \node[right=10pt] at (p2) {$\min$};

\end{tikzpicture}
\caption{$\chi$ becomes too negative.}
\label{fig:unambig to bounded gap negative case}
\end{subfigure}

\begin{subfigure}[b]{0.47\textwidth}
\centering
\begin{tikzpicture}[
    >=stealth,
    thick,
    every node/.style={inner sep=1pt},
    dot/.style={circle,draw,inner sep=1pt},
    mindot/.style={circle,draw,double,inner sep=1pt}
]

  \node[dot]   (q0) at (0,0)   {$q_0$};
  \node[dot]   (q1) at (3.5,-0.6)   {$q_1$};
  \node[dot]   (p1) at (3.5,1.4) {$p_1$};
  \node[mindot] (p2) at (6,0)   {$p_2$};
  \node[mindot] (q2) at (6,0.7)   {$q_2$};
  \node[mindot] (q2b) at (7,-0.8)   {$q'_2$};

  \draw[->,dashed] (q0) -- node[above,pos=0.7] {$\chi$} (q1);
  \draw[->] (q0) -- node[above,pos=0.7] {$\rho$} (p1);
  \draw[dotted] (q0) -- node[above,pos=0.7] {$\chi$} (p2);

 \draw[<->,blue] (q1) -- node[fill=white,
    inner sep=2pt
] {$>\!B$} (p1);

  \draw[->] (p1) -- (p2);
  \draw[dashed,->] (q1) -- (q2);
  \draw[dashed,->] (p1) -- (q2b);

  \draw[decorate,decoration={brace,mirror,amplitude=4pt}]
        (0,-1) -- node[below=4pt] {$x$} (3.5,-1);
  \draw[decorate,decoration={brace,mirror,amplitude=4pt}]
        (3.5,-1) -- node[below=4pt] {$y$} (6,-1);

  \node[right=10pt] at (p2) {$\min$};

\end{tikzpicture}
\caption{$\rho$ decreases too much after $p_1$, leading to a negative cycle.}
\label{fig:unambig to bounded gap positive case}
\end{subfigure}
\caption{Contradiction scenarios for \cref{sec:unambig implies bounded gap}. In \cref{fig:unambig to bounded gap negative case} the run $\chi$ becomes too negative, so that a short suffix induces a negative run to $q'_2$. In \cref{fig:unambig to bounded gap positive case}, the run $\rho$ decreases too much between $p_1$ and $p_2$, causing a negative cycle, which again leads to a negative run to $q'_2$.}
\label{fig:unambiguisable to bounded gaps}
\end{figure}

\subsection{Bounded Gaps $\implies$ $\cA$ is Unambiguisable}
\label{sec:bounded gap implies unambig}
\newcommand{\Bwindows}{B\textsf{-Win}}
\newcommand{\canrho}{\rho_{\preceq}}
Assume that $\cA$ has gaps bounded by $B$. Intuitively, we would want an equivalent unambiguous WFA to track the gaps from the minimal run, and disregard runs that go higher than $B$ above it. Indeed, this is the characterisation in the determinisable case. Unfortunately, the minimal \emph{prefix} of a run might not extend to a minimal run, as it may get ``cut''. We therefore need to nondeterministically guess the minimal run. Then, however, we lose unambiguity. 

To overcome this, we define a notion of \emph{canonical minimal run}, and show that it is unique and can be guessed and verified using an unambiguous WFA. 
We illustrate this in \cref{xmp:unambiguisation and canonical} below.

Fix some arbitrary linear order $\preceq$ on the states $Q$. We think of this order as a priority, where higher priority states are better. Consider a word $w=\sigma_1\cdots \sigma_n$ accepted by $\cA$ and let $\Upsilon$ be the set of minimal-weight accepting runs of $\cA$ on $w$. 
Since $\Upsilon$ is finite, we can denote its runs by $\{\rho^i=q_0^i,\ldots,q_n^i\mid 1\le i\le m\}$ for some $m$. 
We now describe a procedure for culling runs from $\Upsilon$ until we are left with a single run. 

Consider the sequence $\Upsilon_{n+1}\supset \Upsilon_n\supset\ldots \supset \Upsilon_0$ defined inductively (from $n$ to $0$) as follows.
\begin{itemize} 
    \item $\Upsilon_{n+1}=\Upsilon$.
    \item For $0\le k\le n$ we define $\Upsilon_{k}=\{\rho^i\mid \rho^i\in \Upsilon_{k+1}\wedge q_{k}^{i'}\preceq q_k^{i} \text{ for every }i' \text{ such that }\rho^{i'}\in \Upsilon_{k+1}\}$.
\end{itemize}
Intuitively, we consider the set of all minimal runs on $w$, and start scanning them from the end backwards. We first remove all runs for which $q^i_n$ is not $\preceq$-maximal. Then, from the remaining runs (if there are more than one), we keep only runs where $q^i_{n-1}$ is $\preceq$-maximal, and so on.

Note that for $0\le k\le n$, the runs in  $\Upsilon_k$ are all identical from index $k$. Therefore, $\Upsilon_0$ has a single run $\canrho$, which we dub the \emph{canonical run on $w$}.
By definition, $\canrho$ is a minimal run of $\cA$ on $w$. Also, since $\preceq$ is a linear order, the procedure above is deterministic, meaning that $\canrho$ is uniquely defined given $\preceq$.

We can now construct an equivalent unambiguous WFA $\cU$. Intuitively, upon reading a word $w$, the WFA $\cU$ attempts to track the canonical minimal run of $\cA$ on $w$. To do so, $\cU$ keeps track of all the runs in a window of weight $\pm B$ around a (nondeterministically chosen) state $q$. If all the runs stay close to $q$, then all the runs are tracked. However, once a run becomes too high or too low, the window tracks it as $\infty$ or $-\infty$, respectively.
Then, when the word ends, if the current state $q$ is accepting, has minimal weight in the window (in particular there are no accepting runs with weight $-\infty$) and $q$ has maximal priority, then this state accepts.

The crux of the construction is that due to the gap property, if we indeed track the canonical run, then all other accepting states end within its $\pm B$ window, with higher weight or lower priority. In addition, other accepting runs that do not become minimal do not yield accepting runs of $\cU$, since their windows invariably ``believe'' that the canonical run has lower weight or higher priority, and therefore are not marked as accepting.

\begin{example}
    \label{xmp:unambiguisation and canonical}
    Consider the WFA in \cref{fig:unambiguisable WFA}, with the ordering $q_0\preceq q_1\preceq\ldots \preceq q_5$. There are two runs on the word $aaa$, both minimal: $\rho_1=q_0,q_1,q_3,q_5$ and $\rho_2=q_0,q_2,q_4,q_5$. 
    We then have $\Upsilon_4=\Upsilon_3=\{\rho_1,\rho_2\}$. Since $q_3\preceq q_4$, we have $\Upsilon_2=\Upsilon_1=\Upsilon_0=\{\rho_2\}$, which is the canonical run.

    An equivalent unambiguous WFA $\cU$ is in \cref{fig:unambiguisation construction}. The top run tracks ``windows'' around $\rho_1$, reflecting the relative weight of each state from the corresponding state in $\rho_1$. The bottom run similarly tracks $\rho_2$. Notice, however, that from 
    $q_3,\left(\begin{aligned} q_3,&0\\[-3pt] q_4,&2 \end{aligned}\right)$ 
    there is no transition to $q_5$. The reason is that this state ``believes'' that $q_4$, which currently has minimal weight $2$ above $q_3$, can also reach $q_5$ with the same weight (namely $1$) as that from $q_3$, but since $q_4$ has higher priority, this disables the transition from $q_3$. In the formal construction this is enforced using a \emph{consistency check}.
\end{example}

\begin{figure}[ht]
\centering
\begin{subfigure}[b]{0.33\textwidth}
\centering
\begin{tikzpicture}[
  ->, >=stealth',
  auto,
  semithick,
  node distance=1.5cm,
  every state/.style={circle, draw, minimum size=16pt, inner sep=1pt}
  ]

  \node[state] (q0) at (0,0) {$q_0$};
  \node[state] (q1) at (1.5,0.9) {$q_1$};
  \node[state] (q3) at (3,0.9) {$q_3$};
  \node[state] (q2) at (1.5,-0.9) {$q_2$};
  \node[state] (q4) at (3,-0.9) {$q_4$};
  \node[state,accepting] (q5) at (4.5,0) {$q_5$};

  \node[coordinate] (init) at (-0.7,0) {};
  \draw[->] (init) -- (q0);

  \draw (q0) -- node[above,pos=0.3]  {$a,1$}  (q1);
  \draw (q1) -- node[above]       {$a,-2$} (q3);
  \draw (q3) -- node[above] {$a,1$}  (q5);

  \draw (q0) -- node[below=5pt,pos=0.3]  {$a,-1$} (q2);
  \draw (q2) -- node[below]       {$a,2$}  (q4);
  \draw (q4) -- node[below=5pt,pos=0.6] {$a,-1$} (q5);
\end{tikzpicture}
\caption{An \unable WFA $\cA$.}
\label{fig:unambiguisable WFA}
\end{subfigure}

\begin{subfigure}[b]{0.63\textwidth}
\centering
\begin{tikzpicture}[
    ->, >=stealth',
    auto,
    semithick,
    node distance=2.2cm,
    every node/.style={font=\small},
    conf/.style={draw, rounded corners=6pt, inner sep=4pt}
  ]

  \node[conf] (q0) at (0,0) {$q_0, \left(q_0,0\right)$};
  \node[coordinate] (init) at (-1.2,0) {};
  \draw[->] (init) -- (q0);

  \node[conf] (S10) at (2.8,0.9)
    {$q_1,\left(\begin{aligned} q_1,&0\\[-3pt] q_2,&-2 \end{aligned}\right)$};
  \node[conf] (S1m1) at (2.8,-0.9)
    {$q_2,\left(\begin{aligned} q_1,&2\\[-3pt] q_2,&0 \end{aligned}\right)$};

  \draw (q0) -- node[above] {$a,1$}  (S10);
  \draw (q0) -- node[below=5pt] {$a,-1$} (S1m1);

  \node[conf] (S20) at (6,0.9)
    {$q_3,\left(\begin{aligned} q_3,&0\\[-3pt] q_4,&2 \end{aligned}\right)$};
  \node[conf] (S2m1) at (6,-0.9)
    {$q_4,\left(\begin{aligned} q_3,&-2\\[-3pt] q_4,&0 \end{aligned}\right)$};

  \draw (S10)  -- node[above] {$a,-2$} (S20);
  \draw (S1m1) -- node[below] {$a,2$}  (S2m1);

  \node[conf] (Sf) at (9,0)
    {$q_5,\left(q_5,0\right)$};  

  \draw (S2m1) -- node[below] {$a,-1$} (Sf);

\end{tikzpicture}

\caption{The equivalent Unambiguous WFA $\cU$.}
\label{fig:unambiguisation construction}
\end{subfigure}
\caption{\cref{fig:unambiguisable WFA} has gaps bounded by $2$. In \cref{fig:unambiguisation construction} we demonstrate the construction of \cref{sec:bounded gap implies unambig}, with the order $q_0\preceq q_1\preceq q_2\preceq q_3 \preceq q_4\preceq q_5$. Crucially, note that the transition from $q_3$ to $q_5$ is removed in $\cU$. This is due to the consistency check, and since $q_3\preceq q_4$.}
\label{fig:unambiguisation from bounded gaps example}
\end{figure}

The precise construction and correctness proof are given in \cref{apx:sec:bounded gap implies unambig}.

\section{Unambiguisability and Determinisability}
\label{sec:unabmig and determin}
In this section we use our characterisation of \unable WFAs to obtain our main contribution -- a reduction from the unambiguisability problem to the determinisation problem.
The latter was recently shown to be decidable in~\cite{almagor2026determinization}.

\subsection{A Gap Characterisation for Determinisability}
\label{sec:determinisability gap characterisation}
We start by recalling a gap characterisation for determinisable WFAs, captured by \Dtype gap witnesses (where $\mathfrak{D}$ stands for ``$\mathfrak{D}$eterministic''). See \cref{fig:dtype gap witness} for a depiction.
\begin{definition}[\Dtype $B$-Gap Witness]
\label{def: det B gap witness}
    For $B\in \bbN$, a \emph{\Dtype $B$-gap witness over an alphabet $\Sigma$} consists of a pair of words $x,y\in \Sigma^*$ and states $q_1,p_1\in Q$, $p_2\in F$ such that there exist runs $\rho:q_0\runsto{x}p_1\runsto{y}p_2$ and $\chi:q_0\runsto{x}q_1$ and the following holds.
    \begin{itemize} 
        \item $\minweight(q_0\runsto{x}Q)=\weight(\chi)$, i.e. $\chi:q_0\runsto{x}q_1$ is a minimal-weight run on $x$ (not necessarily accepting).
        \item $\minweight(q_0\runsto{xy}F)=\weight(\rho)$, i.e., $\rho$ is a minimal accepting run on $xy$.
        \item $\weight(\rho[x])-\weight(\chi[x])> B$, i.e., after reading $x$ and reaching states $p_1,q_1$, the run $\rho$ is at least $B$ above the minimal run $\chi$.
    \end{itemize}
\end{definition}
We say that a WFA $\cA$ has \emph{\Dtype gaps bounded by $B$} if there are no \Dtype $B+1$ gap witnesses. 
A folklore result (see~\cite{almagor2026determinization} for a precise proof) states that bounded \Dtype gap witnesses characterise determinisability, as follows.
\begin{theorem}
    \label{thm:det iff bounded gap}
    Consider a trim WFA $\cA$, then $\cA$ is determinisable if and only if there exists $B\in \bbN$ such that $\cA$ has \Dtype gaps bounded by $B$.
 \end{theorem}
\begin{remark}[\Utype vs.\ \Dtype gap witnesses]
\label{rmk:Utype vs Dtype}
There is an obvious similarity between \Dtype witnesses (\cref{def: det B gap witness}) and \Utype witnesses (\cref{def:U type B gap witness}), and understanding the differences between the two is key to our proof. 
First, notice that every \Utype $B$-gap witness is in particular a \Dtype $B$-gap witness. Indeed, being a \Dtype witness is a weaker requirement, so that the absence of \Dtype $B$-gap witnesses is a stronger requirement implying determinisability rather than \unability.

For the converse, a \Dtype $B$-gap witness is \emph{not} a \Utype $B$-gap witness when the run $\chi:q_0\runsto{x}q$ cannot be continued to an accepting run on $xy$ (and this is the only difference). 

\begin{figure}[ht]
\centering
\begin{subfigure}[b]{0.45\textwidth}
\centering
\begin{tikzpicture}[
    >=stealth,
    thick,
    every node/.style={inner sep=1pt},
    dot/.style={circle,draw,inner sep=1pt},
    mindot/.style={circle,draw,double,inner sep=1pt}
]

  \node[dot]   (a1) at (0,0)   {$q_0$};
  \node[dot]   (b1) at (3.5,-0.3)   {$q_1$};
  \node[dot]   (c1) at (3.5,1.5) {$p_1$};
  \node[mindot] (m1) at (6,0)   {$p_2$};
  \node[mindot] (m2) at (6,1)   {$q_2$};

  \draw[->,dashed] (a1) -- node[above,pos=0.7] {$\chi$} (b1);
  \draw[->] (a1) -- node[above,pos=0.7] {$\rho$} (p1);

 \draw[<->,blue] (b1) -- node[fill=white,
    inner sep=2pt
] {$>\!B$} (c1);

  \draw[->] (c1) -- (m1);
  \draw[dashed,->, line width=2pt] (b1) -- (m2);

  \draw[decorate,decoration={brace,mirror,amplitude=4pt}]
        (0,-0.7) -- node[below=4pt] {$x$} (3.5,-0.7);
  \draw[decorate,decoration={brace,mirror,amplitude=4pt}]
        (3.5,-0.7) -- node[below=4pt] {$y$} (6,-0.7);

  \node[right=10pt] at (m1) {$\min$};

\end{tikzpicture}
\caption{\Utype $B$-gap witness}
\label{fig:utype gap witness}
\end{subfigure}
\hfill
\begin{subfigure}[b]{0.47\textwidth}
\centering
\begin{tikzpicture}[
    >=stealth,
    thick,
    every node/.style={inner sep=1pt},
    dot/.style={circle,draw,inner sep=1pt},
    mindot/.style={circle,draw,double,inner sep=1pt}
]

  \node[dot]   (a1) at (0,0)   {$q_0$};
  \node[dot]   (b1) at (3.5,-0.3)   {$q_1$};
  \node[dot]   (c1) at (3.5,1.5) {$p_1$};
  \node[mindot] (m1) at (6,0)   {$p_2$};

  \draw[dashed,->] (a1) -- node[above,pos=0.7] {$\chi$} (b1);
  \draw[->] (a1) -- node[above,pos=0.7] {$\rho$} (c1);

 \draw[<->,blue] (b1) -- node[fill=white,
    inner sep=2pt
] {$>\!B$} (c1);

  \draw[->] (c1) -- (m1);

  \draw[decorate,decoration={brace,mirror,amplitude=4pt}]
        (0,-0.7) -- node[below=4pt] {$x$} (3.5,-0.7);
  \draw[decorate,decoration={brace,mirror,amplitude=4pt}]
        (3.5,-0.7) -- node[below=4pt] {$y$} (6,-0.7);

  \node[right=10pt] at (m1) {$\min$};

\end{tikzpicture}
\caption{\Dtype $B$-gap witness}
\label{fig:dtype gap witness}
\end{subfigure}
\caption{$B$-gap witness. The vertical height represents the weight. After reading $x$, the run $\chi$ is minimal, and $\rho$ is far above it. Upon reading $y$, $\rho$ continues to become a minimal run. In \Utype witnesses, $\chi$ must also continue to become accepting. In \Dtype, there is no requirement on $\chi$ (but $q_1$ can reach $F$ via \emph{some} word, since the automaton is trim).} 
\label{fig:gap witness}
\end{figure}

\end{remark}

\subsection{Reducing Unambiguisability to Determinisability}
\label{sec:reduction}
We now turn to our main result.
\begin{theorem}
\label{thm:main reduction}
    The Unambiguisability problem is reducible to the Determinisability problem.
\end{theorem}
Before delving into the proof, we give some intuition.
Consider a WFA $\cA$. We wish to construct from $\cA$ a WFA $\cB$ such that $\cA$ is \unable if and only if $\cB$ is determinisable. 
In light of \cref{rmk:Utype vs Dtype}, we actually aim that every \Utype gap $B$-witness for $\cA$ induces a \Dtype $B$-gap witness for $\cB$, and that $\cB$ does not have any \Dtype $B$-gap witnesses that are not also \Utype. 
The former requirement is easy -- all we need to do is maintain enough of the structure of $\cA$ so as not to cause too much havoc (i.e., maintain the \Utype witnesses, which are already also \Dtype).

Making sure there are no further \Dtype witnesses in $\cB$ is the challenging part. To achieve this, we essentially ``prune'' the runs of $\cA$ as follows. At each state of $\cB$, we maintain a \emph{commitment}, which is a function $f$ that describes for every state $q\in Q$ whether $q$ is going to reach the accepting state ($f(q)=\keep$), whether $q$ is going to reach some states, but not the accepting state ($f(q)=\trim$), or whether $q$ is unreachable ($f(q)=\bot$). Then, with each letter we also receive an \emph{update} function $\alpha$ which states for every \emph{transition} whether it is along an accepting run ($\keep$), only along non-accepting runs ($\trim$), or unavailable ($\bot$). The commitments are updated deterministically, and must correctly follow the run DAG of $\cA$ on the word. Here, the run DAG is the layered graph whose $i$-th layer contains the states reachable after the first $i$ letters, and whose edges are the transitions used between consecutive layers.
The idea is then that in a \Dtype witness in $\cB$, the ``lower'' run $\chi$ on $x$ must be extendable to an accepting run on $xy$, since the updates given by $y$ dictate that there is such an extension. Thus, we can convert a \Dtype witness to a \Utype one.

We prove \cref{thm:main reduction} in the remainder of the section, starting with the construction. 
\subsubsection{The Reduction Construction}
\label{cref:reduction unambig to det construction}
Consider a WFA $\cA=\tup{Q,\Sigma,q_0,\Delta,F}$. We assume (based on \cref{rmk: initial and final weights}) that $F=\{q_\fin\}$ is the unique accepting state of $\cA$.
We obtain from $\cA$ a WFA $\cB=\tup{S,\Gamma,s_0,\Lambda,G}$ such that $\cA$ is \unable if and only if $\cB$ is determinisable. 
We start with some auxiliary definitions before describing $\cB$. 
Consider the set $\commit=\{\bot,\trim,\keep\}^Q$. We refer to each $f\in \commit$ as a \emph{commitment}, which intuitively prescribes to each state whether it is unreachable ($\bot$), reachable and is along an accepting run ($\keep$) or reachable but not along an accepting run ($\trim$).

Next, consider the set $\update=\{\bot,\trim,\keep\}^{Q\times Q}$. We refer to each $\alpha\in \update$ as an \emph{update}, which intuitively prescribes to each $p,q\in Q$ whether the transition from $p$ to $q$ is not available ($\bot$), is available along an accepting run ($\keep$) or is available but not along an accepting run ($\trim$). We abbreviate and write $p\bot q\in \alpha$, $p\keep q\in \alpha$, $p\trim q\in \alpha$ to signify these three cases, respectively. We illustrate the construction in \cref{fig:reduction example}.
\begin{figure}[ht]
\centering
\begin{subfigure}[b]{0.3\textwidth}
\centering
\begin{tikzpicture}[
    ->, >=stealth',
    auto,
    semithick,
    every state/.style={circle, draw, minimum size=16pt, inner sep=1pt}
  ]

  \node[state,initial, initial text={}] (q) at (0,0) {$q$};
  \node[state] (p) at (2,0.9) {$p$};
  \node[state] (r) at (2,-0.9) {$r$};
  \node[state, accepting] (s) at (4,0) {$s$};

  \draw (q) -- node[above left] {$a,0$} (p);
  \draw (q) -- node[below left] {$a,0$} (r);

  \draw (p) edge[loop above, looseness=7, in=120, out=60] node[pos=0.25, right] {$b,1$} (p);
  \draw (r) edge[loop below, looseness=7, in=300, out=240] node[pos=0.75, right] {$b,0$} (r);

  \draw (p) -- node[above right] {$c,0$} (s);
  \draw (r) -- node[below right] {$d,0$} (s);

\end{tikzpicture}

\caption{An unambiguous WFA $\cA$.}
\label{fig:unambiguous for reduction}
\end{subfigure}

\begin{subfigure}[b]{0.6\textwidth}
\centering
\begin{tikzpicture}[
    ->, >=stealth',
    auto,
    semithick,
    node distance=3cm,
    every node/.style={font=\small},
    box/.style={draw, rounded corners=6pt, inner sep=4pt}
]

\node[box,initial,initial text={}] (S0) at (0,0) {$q,\; q\!\keep$};

\node[box] (T1) at (3.5, 0.7)
  {$p,\begin{aligned}
      p &\keep \\[-5pt]
      r &\trim
    \end{aligned}$};

\node[box] (T2) at (3.5, -0.7)
  {$r,\begin{aligned}
      p &\trim \\[-5pt]
      r &\keep
    \end{aligned}$};

\node[box,accepting] (S3) at (7.5, 0)
  {$s,\; s\!\keep$};

  \draw (S0) -- node[pos=0.25, above=7pt]
    {$\left( a, \begin{aligned}
      q &\keep p \\[-5pt]
      q &\trim r
    \end{aligned}\right),\; 0$}
    (T1);

  \draw (S0) -- node[pos=0.25, below=7pt]
    {$\left( a, \begin{aligned}
      q &\trim p \\[-5pt]
      q &\keep r
    \end{aligned}\right),\; 0$}
    (T2);

    \draw (T1) edge[loop above]
      node[pos=0.75, right]
      {$\left( b, \begin{aligned}
        p &\keep p \\[-5pt]
        r &\trim r
      \end{aligned}\right),\; 1$} (T1);

  \draw (T2) edge[loop below]
    node[pos=0.25, right]
      {$\left( b, \begin{aligned}
        p &\trim p \\[-5pt]
        r &\keep r
      \end{aligned}\right),\; 0$} (T2);

  \draw (T1) -- node[pos=0.75, above=7pt]
    {$(c, p\!\keep s),\; 0$}
    (S3);

  \draw (T2) -- node[pos=0.75,below=7pt]
    {$(d, r\!\keep s),\; 0$}
    (S3);

\end{tikzpicture}

\caption{The (deterministic) reduction output $\cB$ WFA.}
\label{fig:deterministic reduction output}
\end{subfigure}

\caption{The reduction of \cref{thm:main reduction}. The WFA $\cA$ is \unable (as it is already unambiguous). The reduction output $\cB$ adds the commitments and updates to the transitions. $\bot$ markings are omitted for clarity. 
For example, in order to take the $q\to p$ transition, the commitment specifies that $p$ leads to an accepting state, \emph{and} that $r$ does not, thus fixing an explicit run DAG.
Note that $\cB$ is determinisable (as it is already deterministic).} 
\label{fig:reduction example}
\end{figure}

We now turn to define $\cB$. The states are  $S=Q\times \commit$. That is, each state is a pair $(q,f)$ where $q\in Q$ and $f\in \commit$. 
The alphabet is $\Gamma=\Sigma\times \update$. That is, at each transition $\cB$ reads a letter $\sigma\in \Sigma$ as well as an update $\alpha\in \update$. 
The initial state is $s_0=(q_0,f_0)$ where $f_0\in \commit$ is the commitment $f_0(q_0)=\keep$ and $f_0(p)=\bot$ for every $p\neq q_0$.
The accepting states are 
\[G=\{(q_\fin, f_\fin)\mid f_\fin(q_\fin)=\keep\wedge f_\fin(p)\neq \keep \text{ for every } p\neq q_\fin\}\]
The transitions $\Lambda$ are as follows.
Consider two states $(q,f),(p,g)\in S$ and a letter $(\sigma,\alpha)\in \Gamma$. We have $((q,f),(\sigma,\alpha),c,(p,g))\in \Lambda$ if and only if the following consistency conditions hold.
\begin{itemize} 
      \item \textbf{$\Delta$-consistency:} $(q,\sigma,c,p)\in \Delta$ (i.e., the projection to $\cA$ is a valid transition with the same weight).
    \item \textbf{Update consistency:} for every $r,t\in Q$ we have 
    $(r,\sigma,\infty,t)\in \Delta$ if and only if $r\bot t\in \alpha$. 
    Equivalently, $\minweight(r\runsto{\sigma}t)\neq \infty$ if and only if $r\keep t\in \alpha$ or $r\trim t\in \alpha$. That is, the update $\alpha$ correctly reflects the available transitions on $\sigma$, marking them with $\keep$ and $\trim$. 
    Note that this condition depends only on the letter $(\sigma,\alpha)$, not on the states.
    \item \textbf{Outgoing consistency:} for every $r\in Q$ we have:
    \begin{itemize}
        \item If $f(r)=\trim$ then for every $r'\in Q$ we have $r\keep r'\notin \alpha$ (i.e., outgoing edges from $\trim$ states are marked $\bot$ or $\trim$).
        \item If $f(r)=\keep$ then there exists $r'\in Q$ such that $r\keep r'\in \alpha$.
    \end{itemize}
    
    \item \textbf{Incoming consistency:} for every $r'\in Q$ we have:
    \begin{itemize}
        \item $g(r')=\keep$ if there exists $r\in Q$ such that $f(r)=\keep$ and $r\keep r'\in \alpha$, and for every $r\in Q$ if $r\trim r'\in \alpha$ then $f(r)= \bot$.
        \item $g(r')=\trim$ if there exists $r\in Q$ such that $f(r)\neq \bot$ and $r\trim r'$.
    \end{itemize}
\end{itemize}




Intuitively, at each state $(q,f)$ $\cB$ commits to certain states (of $\cA$) leading to $q_\fin$, and others not leading to $q_\fin$. Then, $\cB$ reads a letter $(\sigma,\alpha)$ where $\alpha$ describes exactly the available transitions on $\sigma$. The state component $q$ is updated nondeterministically according to $\sigma$ in $\cA$. The commitment is updated \emph{deterministically} according to $f$ and $\alpha$: outgoing consistency checks that the old commitment can be extended correctly, i.e., that $\keep$ transitions reach some $\keep$ state and $\trim$ states do not admit $\keep$ transitions, and incoming consistency uniquely determines whether each state in the next layer is marked $\keep$, $\trim$, or $\bot$.

At a higher-level, $\cB$ essentially reads a word along with a specific run-DAG on it, where some runs are marked ``trimmed'' ($\trim$), which intuitively means that they do not lead to accepting states.
The detailed correctness proof is in~\cref{apx:reduction correctness}. We outline the ideas here.

The first step is to show a correspondence between $\cA$ and $\cB$. For a word $w\in \Gamma^*$, we denote by $w|_\Sigma$ its projection on $\Sigma^*$. Similarly, for a run $\rho$ of $\cB$ we denote by $\rho|_Q$ its projection on $Q$. We show in \cref{prop: redcution correspondence} that a run of $\cB$ can be projected to a run of $\cA$ by removing the commitments and updates, and conversely -- a run of $\cA$ can be lifted to a run of $\cB$ by providing exactly the correct commitments and updates from the run DAG of $\cA$. Moreover, this correspondence maintains the weights of the runs.
We now proceed to show correctness. 

 \paragraph*{$\cA$ is Not Unambiguisable $\implies$ $\cB$ is Not Determinisable}
We prove this direction via the gap characterisation (\cref{thm:unambig iff bounded U gap,thm:det iff bounded gap}). Specifically, we prove that if there is a \Utype $B$-gap witness in $\cA$, then there is a \Dtype $B$-gap witness in $\cB$. This follows easily from the correspondence above: any \Utype $B$-gap witness induces a \Utype $B$-gap witness in $\cB$, and this is in particular a \Dtype witness (\cref{rmk:Utype vs Dtype}).

 \paragraph*{$\cA$ is Unambiguisable $\implies$ $\cB$ is Determinisable}
We turn to the ``hard'' direction. We again use gap witnesses, this time showing that every \Dtype $B$-gap witness in $\cB$ induces a \Utype $B$-gap witness in $\cA$. We first assume without loss of generality that $\cB$ is trim. However, we remark that this is an important assumption that is treated carefully in the proof.
Consider therefore a \Dtype $B$-gap witness $xy$ in $\cB$, with the corresponding runs $\rho:(q_0,f_0)\runsto{x}(p,f_p)\runsto{y}(q_\fin,g)$ and $\chi:(q_0,f_0)\runsto{x}(q,f_q)$ (where $\chi$ is minimal on $x$, and $\rho$ is minimal on $xy$).
We claim that $xy|_\Sigma$ is a \Utype $B$-gap witness in $\cA$. Note that this almost holds by \cref{rmk:Utype vs Dtype}, and the only thing left to show is that $\chi|_Q$ can be extended to some run on $y|_\Sigma$. 

This is where the consistency requirements in the construction of $\cB$ come into play. First, we observe that $f_p(p)=f_q(q)=\keep$. Indeed, if, for example, $f_q(q)\neq \keep$, then the commitment reached so far forces that $q$ cannot proceed to an accepting state. But then the state $(q,f_q)$ cannot reach any accepting state on any word, as any accepting continuation would violate the outgoing and incoming consistency requirements. This contradicts the assumption that $\cB$ is trim. Thus, since $\cB$ is trim, such a state cannot occur in a witness.

Next, observe that since the commitment component is updated deterministically, we have $f_p=f_q$. In particular, $f_p(q)=f_p(p)=\keep$. Intuitively, this means that all the runs have the same opinion on whether $q$ reaches an accepting state, and this opinion is $\keep$. We can then inductively follow the update consistency on $y$ from $(q,f(q))$, and we are guaranteed that there is some run from $q$ on $y$ that reaches an accepting state, as required.

We remark that a concerned reader may wonder why we place so much emphasis on a simple property such as being trim. This is in a way the key to the proof: the construction of $\cB$ is such that if a state is not trimmed, then it can reach an accepting state on \emph{every} upcoming suffix (provided there is at least one other run that can read this suffix). \hfill \qed

This concludes the proof of \cref{thm:main reduction}. Since determinisability is decidable by~\cite{almagor2026determinization}, we have the following.
\begin{restatable}{corollary}{unambigdecidable}
    \label{cor:unambiguisability is decidable}
    The Unambiguisability problem for WFA is decidable.
\end{restatable}

Finally, we remark that currently there are no known complexity upper bounds for determinisability, and therefore our reduction does not provide complexity bounds either. It should be noted, however, that the reduction has a single-exponential blowup in the state space. Once complexity bounds for determinisation are established, it would be interesting to see if this blowup is necessary, or whether there is a polynomial-time (or indeed -- logspace) reduction.
\begin{theorem}
\label{thm:unambiguisability pspace hard}
    WFA Unambiguisability is PSPACE-hard.
\end{theorem}
\begin{proof}
The PSPACE-hardness proof of determinisation in~\cite[Appendix D]{almagor2026determinization} actually uses \Utype witnesses, not just \Dtype witnesses. It therefore works word-for-word to show that \unability is also PSPACE-hard.
\end{proof}

\section{Minimising Registers in Cost Register Automata}
\label{sec: CRA register minimization}
Tropical cost register automata (CRAs) with linear register updates provide an alternative representation of WFAs, where nondeterminism is captured in the behaviour of several registers, keeping the control deterministic. This view offers a natural measure of nondeterminism by the number of registers needed to capture a function. 
In this section we show that unfortunately, minimising the number of registers is generally undecidable already for inputs with $7$ registers, and therefore also for any larger fixed number of registers. 

A fundamental result in~\cite{alur2013regular} is that CRAs are expressively equivalent to WFAs. In~\cref{apx:CRA} we formally define CRAs, and refine this result by showing that CRAs with $k$ registers are equivalent to WFAs in which, for every word, the maximal number of states simultaneously reachable is $k$. We refer to these as \emph{width-$k$} WFAs.
Thus, our result is the following (see also \cref{cor:CRA register minimisation undecidable}). The proof is in \cref{apx:sec: width reduction undecidable}, and we bring here only the core idea.
\begin{restatable}{theorem}{widthund}
    \label{thm:width reduction undecidable}
    The following problem is undecidable already for $k=7$: given a width $k$ WFA, decide whether there is an equivalent width $k-1$ WFA. 
\end{restatable}

\begin{figure}[ht]
\centering
\begin{tikzpicture}[
    ->, >=stealth',
    auto,
    semithick,
    node distance=1.8cm,
    every state/.style={circle, draw, minimum size=20pt, inner sep=1pt},
    every loop/.style={looseness=4, min distance=25pt}
  ]

  \node[state,initial above,initial text={}] (qa) at (-3,0.3) {$q_a$};
  \node[state,initial above,initial text={}] (q0) at (0,0.3) {$q_0$};

  \node[state] (c1)   at (2,-0.3) {$q_{\Cl{1}}$};
  \node  (dots) at (4.7,-0.2) {$\ldots$};
  \node[state] (c6)   at (6,-0.3) {$q_{\Cl{6}}$};

  \draw[rounded corners] (-1.6,0.9) rectangle (0.6,-0.9);
  \node at (-1.3,0.6) {$\cA$};

  \draw[->] (qa) edge[in=210, out=150, looseness=6]
    node[pos=0.5, left, align=center] {$\sigma,0$} (qa);

  \draw[->, dotted] (q0) edge[in=220, out=270, looseness=7]
    node[pos=0.7,left=-1pt] {$ w,1 $} (q0);

  \draw[->] (q0) to[bend right=10] node[pos=0.6, above] {$\sep,0$} (c1);
  \draw[->] (q0) to[bend right=30] node[pos=0.2, below] {$\sep,0$} (c6);

  \draw[->] (c1) edge[in=-30, out=30, looseness=6]
    node[align=center, right] {$\begin{aligned} 
      \Cl{1},&-1\\[-4pt] 
      \sigma,&0 
    \end{aligned}$} (c1);

  \draw[->] (c6) edge[in=-30, out=30, looseness=6]
    node[align=center, right] {$\begin{aligned} 
      \Cl{6},&-1\\[-4pt] 
      \sigma,&0 
    \end{aligned}$} (c6);

    \node (X1) at ($(c1)+(0,1.3)$) {\Large$\times$};
\draw[->] (c1) -- node[right] {$\xCl{1}$} (X1);

\node (X6) at ($(c6)+(0,1.3)$) {\Large$\times$};
\draw[->] (c6) -- node[right] {$\xCl{6}$} (X6);

\node (Xa) at ($(qa)+(0,-1.1)$) {\Large$\times$};
\draw[->] (qa) -- node[right] {$a$} (Xa);

\end{tikzpicture}
\caption{Reduction idea. The letter $\sigma$ represents all ``non-killing'' letters in the transitions. Killing letters lead to $\times$.
Intuitively, we start either at $q_a$ or at $q_0$ (in $\cA$). In $q_a$ there is a self loop with weight $0$ on everything except $a$. From $q_0$ we can run as $\cA$, but also move to the $\Cl{i}$ components on $\sep$. In state $q_{\Cl{i}}$ we self loop with weight $-1$ on $\Cl{i}$, and with $0$ on all other letters except $\xCl{i}$. There may be a self loop on $q_0$ with the word $w$.}
\label{fig:reduction 7 width core}
\end{figure}

Consider a WFA $\cA$ of width $6$, such that either $\cA(w)\le 0$ for all $w$, or there is a word $w$ such that $\minweight(q_0\runsto{w}q_0)=1$, so that\footnote{This actually requires further assumptions, see the detailed proof.} the weight of $w^n$ is $n$. We want to separate the two cases using a reduction to the width-minimisation problem. 
We obtain from $\cA$ a new WFA $\cA'$ as depicted in \cref{fig:reduction 7 width core}. We add to the alphabet the letters $\{\sep,\Cl{1},\ldots,\Cl{6},a,\xCl{1},\ldots,\xCl{6}\}$ and introduce new states $q_a, q_{\Cl{1}},\ldots, q_{\Cl{6}}$. The behaviour is the following: we start both at $q_a$ and at $q_0$ (in $\cA$). When $\sep$ is read, runs from $q_0$ leave to all the $q_{\Cl{i}}$. There, $q_{\Cl{i}}$ loses $1$ on $\Cl{i}$, and is not allowed to read $\xCl{i}$. In addition component $q_a$ maintains weight $0$ but cannot read $a$.

The intuitive idea is the following. First, we claim that in any WFA equivalent to $\cA'$, the states $q_{\Cl{i}}$ must be tracked separately, leading to width $6$ at least. This is because upon reading the word e.g., $\Cl{1}^{n}\Cl{2}^{2n}\Cl{3}^{3n}\Cl{4}^{4n}\Cl{5}^{5n}\Cl{6}^{6n}$ for large $n$, the gaps between the runs on the $q_{\Cl{i}}$ states are very large. Then, we can use the ``killing letters'' $\xCl{i}$ to eliminate all the runs but one, with a very short suffix (e.g., the suffix $\xCl{1}\xCl{3}\xCl{4}\xCl{5}\xCl{6}$ leaves only the $q_{\Cl{2}}$ run). For large enough $n$, a WFA of width 5 cannot track these values correctly. 

Next, we look at the $q_a$ component. If $\cA(w)\le 0$ for all $w$, the component $q_a$ is redundant, so we have an equivalent width-6 WFA. Otherwise, we claim that any equivalent WFA needs to track $q_a$ separately to all the $q_{\Cl{i}}$. Here the reason is that using $w^n$ we can reach weight $n$ in the $\cA$ component (so $q_a$ is far below it, with weight $0$), but upon reading $a$ the $q_a$ component is killed and the weight suddenly jumps to $n$, which again requires another component in any equivalent WFA. Thus, in this case any equivalent WFA must be of width at least $7$.

In order to turn this construction into a proper reduction we show that we can obtain such a WFA $\cA$ for which the separation described above is undecidable. This uses the upper-boundedness construction of~\cite{Almagor2020Whatsdecidableweighted}, which reduces the $0$-halting problem for two-counter machines to deciding whether the function described by a WFA is bounded from above. The detailed proof in \cref{apx:sec: width reduction undecidable} states the precise properties of that construction and then applies the gadget described here.

\section{Discussion and Future Research}
\label{sec:discussion}
In a nutshell, our work maps out the borders of ``nondeterminism minimisation'' in WFAs, showing on the positive side that \unability is decidable, and on the negative side that reducing the width (equivalently -- minimising the number of registers in a CRA) is undecidable.

Note that our results hold for the closely related settings such as max-plus WFAs over $\bbZ_\infty$ and min-plus or max-plus WFAs $\bbN_\infty$. For max-plus over $\bbZ_\infty$, the setting is completely symmetric -- simply negate all the weights to obtain an analogous min-plus WFA.
For $\bbN_\infty$, note that adding a constant to all the transitions retains unambiguisability and determinisability (for both max-plus and min-plus). Therefore, we can start with a $\mathbb{Z}_\infty$ WFA and increase the weights so that they are all in $\mathbb{N}_\infty$, and the results stil hold. Note that "still holds" means that the general problem and the $\mathbb{N}_\infty$ problem are computationally equivalent. 
The decidability already holds because this is a sub-case.

Our decidability proof of \unability relies on the decidability of WFA determinisability. In particular, the current best known complexity bounds for determinisability are in the 6th level of the fast-growing hierarchy~\cite{almagor2026complexity}. Since our reduction has a single-exponential blowup (and the fast-growing hierarchy is closed under single-exponential blowups) then the same complexity bounds hold for \unability as well. For the lower bound, the best current bound is PSPACE-hardness, as per \cref{thm:unambiguisability pspace hard}.

Two natural questions arise from our research. First, can we decide more relaxed ambiguity? E.g., can we decide if a given WFA has an equivalent 2-ambiguous/finitely ambiguous/polynomially-ambiguous WFA? 
The question of 2-ambiguisability seems very difficult, and currently out of reach. In particular, we do not know of a gap criterion that corresponds to 2-ambiguous WFAs. 
The second question is whether register minimisation becomes decidable for $k<7$, which is perhaps of lesser importance, but it would nonetheless be nice to complete the picture.

In addition, now that some borders on decidability are in place, we can map out fragments, e.g., register minimisation for \emph{copyless} CRAs~\cite{Almagor2020WeakCostRegister}.

\bibliography{main}

\pagebreak
\appendix

\section{No Need for Initial and Final Weights}
\label{sec:no need for init and fin and acc}

\newcommand{\slet}{\mathsf{s}}
\newcommand{\flet}{\mathsf{f}}
\newcommand{\weightif}{\weight^{+\mathsf{if}}}
\newcommand{\minweightif}{\mathsf{m}\weight^{+\mathsf{if}}}
\newcommand{\WFAif}{WFA$_\text{i-f}$\xspace}
A \emph{$(\min,+)$ Weighted Automaton with initial and final weights} (\WFAif for short) is a tuple $\cA= \tup{Q,\Sigma, \init, \Delta ,\fin}$ with the same components as a WFA, except that 
$\init,\fin\in \bbZ_\infty^Q$ are $Q$-indexed vectors denoting for each state its \emph{initial weight} and \emph{final weight}, respectively. 

For a run $\rho:p\runsto{x}q$ in a \WFAif, the sum of its weights along the transitions is denoted $\weight(\rho)$ (similarly to WFAs). We also use the following notation for \WFAif: 
$\weightif(\rho)=\init(p)+\weight(\rho)+\fin(q)$ and $\minweightif(Q_1\runsto{x} Q_2)=\min\{\weightif(\rho)\mid \rho:Q_1\runsto{x}Q_2\}$.

A run $\rho$ is accepting if $\weightif(\rho)<\infty$.
A \WFAif is unambiguous if for every $w$ there is at most one accepting run.

In this section we show that the \unability problem for \WFAif reduces to the \unability problem for our model, which only has one initial state, and has no final weights (and in particular $\init,\fin\in \{0,\infty\}^Q$). In addition, we can also assume a single accepting state.

The intuition for getting rid of the initial and final weights is simple: we add two letters $\{\slet,\flet\}$ to the alphabet, to denote the start and the finish of a word, respectively. Upon reading $\slet$ and $\flet$, the automaton incurs the weights described by the initial and final states, respectively.

\begin{lemma}
    \label{lem:init final weight reduces to no init and final}
    The \unability problem for \WFAif is reducible (in logarithmic space) to the \unability problem for WFAs. Moreover, we can assume the WFA has a single accepting state.
\end{lemma}
\begin{proof}
Consider a \WFAif $\cA= \tup{Q,\Sigma, \init, \Delta ,\fin}$, and assume $\cA$ is trim (otherwise remove states that are not reachable from a state with finite initial weight or not co-reachable from a state with finite final weight). 

We construct a WFA $\cB= \tup{Q\cup \{s_0,s_f\},\Sigma\cup \{\slet,\flet\}, s_0, \eta ,F }$ where $s_0,s_f\notin Q$, $F=\{s_f\}$ and $\slet,\flet\notin \Sigma$. The transitions are defined as follows:
    \begin{itemize}
        \item For $p,q\in Q$ and $\sigma\in \Sigma$ we have $(p,\sigma,c,q)\in \eta$ iff $(p,\sigma,c,q)\in \Delta$.
        \item For $q\in Q$ we have $(s_0,\slet,\init(q),q)\in \eta$ and $(q,\flet,\fin(q),s_f)\in \eta$.
        \item The remaining transitions are with weight $\infty$.
    \end{itemize}
    This construction can clearly be implemented in logarithmic space.
    We prove that $\cA$ is \unable if and only if $\cB$ is \unable.

    Observe that by the construction of $\cB$, we have for every word $w\in \Sigma^*$ that
    \begin{equation}
    \label{eq:init fin weight equiv}
    \cA(w)=\cB(\slet \cdot w\cdot  \flet)    
    \end{equation}

     \paragraph*{If $\cA$ is \unable, then $\cB$ is \unable}
    Assume $\cA$ is \unable, and let $\cU$ be an equivalent unambiguous \WFAif. Apply the construction above to $\cU$ and obtain a WFA $\cU'$. By \cref{eq:init fin weight equiv} and the equivalence of $\cA$ and $\cU$, for every word $w\in \Sigma^*$ we have
    \[{\cB}(\slet w\flet)={\cA}(w)={\cU}(w)={\cU'}(\slet w \flet)\]
    Moreover, if $w\notin \slet \cdot \Sigma^*\cdot \flet$, then ${\cB}(w)=\infty={\cU'}(w)$.
    It follows that $\cB$ is equivalent to $\cU'$. It remains to show that $\cU'$ is unambiguous. 
    Assume by way of contradiction that there is a word on which $\cU'$ has at least two accepting runs $\rho,\chi$. It must therefore hold that this word is of the form $\slet w \flet$ with $w\in \Sigma^*$. Then, $\rho$ and $\chi$ can be written as:
    \[
    \begin{split}
    \rho:s_0\runsto{\slet}q_1\runsto{w}q_n\runsto{\flet}s_f\\
    \chi:s_0\runsto{\slet}p_1\runsto{w}p_n\runsto{\flet}s_f\\
    \end{split}   
    \]
    Since $\chi\neq \rho$, it follows that there exists some $1\le i\le n$ such that $q_i\neq p_i$ (otherwise the runs are identical, since there are no two transitions with different weights between the same two states). Moreover, we have both $\fin(q_n)<\infty$ and $\fin(p_n)<\infty$ (in $\cU$), otherwise there would not be a transition to $s_f$.
    But then $q_1\runsto{w}q_n$ and $p_1\runsto{w}p_n$ are two accepting runs of $\cU$ on $w$, contradicting the fact that $\cU$ is unambiguous.

     \paragraph*{If $\cB$ is \unable, then $\cA$ is \unable}
    Assume $\cB$ is \unable and let $\cU=\tup{Q_\cU,\Sigma\cup\{\slet,\flet\}, q_0, \Delta_\cU,F}$ be an equivalent unambiguous trim WFA. 
    Note that we can assume $q_0$ has no incoming transitions (otherwise a word with more than one $\slet$ can have finite weight, unless all words have cost $\infty$, which is a degenerate case).
    Similarly, we can assume that once $\flet$ is read, a unique state $F=\{q_f\}$ is reached, from which there are no outgoing transitions (indeed, no final weights can be accumulated upon leaving $q_f$). We can assume $q_0\neq q_f$, otherwise the only accepted word in $\cA$ is $\epsilon$, which is again degenerate (since $q_f$ has no outgoing transitions).

    Define a deterministic \WFAif  $\cU'=\tup{Q'_{\cU},\Sigma, \init', \Delta'_\cU ,\fin'}$ as follows. The states are $Q'_\cU=Q_\cU\setminus \{q_0,q_f\}$.

    For every $q\in Q'_\cU$ define $\init'(q)=c$ where $c\in \bbZ$ is such that $(q_0,\slet,c,q)\in \Delta_{\cU}$, i.e., each state starts with the weight reached by reading $\slet$.
    Similarly, for the final vector, for every $q\in Q'_\cU$  set $\fin'(q)=c$ where $(q,\flet,c,q_f)\in \Delta_{\cU}$.
    
    For the remaining transitions, we have that $\Delta'_\cU=\Delta_{\cU}\cap (Q'_\cU\times \Sigma\times \bbZinf\times Q'_\cU)$ (i.e., we keep only transitions on $Q'_\cU$ and over $\Sigma$).

    We claim that $\cU'$ is an unambiguous \WFAif equivalent to $\cA$. Starting with the latter, we claim that ${\cU'}(w)={\cA}(w)$ for every $w\in \Sigma^*$. 
    By \cref{eq:init fin weight equiv}, it is enough to prove that ${\cU'}(w)=\cB(\slet\cdot w\cdot \flet)$, but the latter is immediate from the construction, since reading $\slet$ is simulated by the initial weights, and reading $\flet$ by the final weights. 

    Finally, we claim that $\cU'$ is unambiguous. Indeed, assume by way of contradiction that there is a word $w\in \Sigma^*$ on which $\cU'$ has at least two accepting runs $\rho,\chi$. We can therefore obtain two accepting runs of $\cU$ on $\slet w\flet$ by starting from $s_0$ and then following $\rho$ and $\chi$, and finally using $\flet$. The construction of $\cU'$ guarantees these are valid runs. This is a contradiction to the unambiguity of $\cU$.    
\end{proof}

\section{Basic Properties of Unambiguous WFAs}
\label{apx:basic properties of unambig}
We mention two basic results about unambiguous WFAs that are useful in the proof. The first concerns the \emph{negation} of WFAs. Consider a WFA $\cA=\tup{Q,\Sigma,q_0,\Delta,F}$. We obtain a new WFA denoted $\cA^-=\tup{Q,\Sigma,q_0,\Delta',F}$ by defining $\Delta'=\{(q,\sigma,-c,q')\mid (q,\sigma,c,q')\in \Delta\}$, i.e., we negate all the weights in $\cA$. In general, not much can be said about the functions described by $\cA^-$. However, if $\cA$ is unambiguous, then so is $\cA^-$, and for every accepted word, the value of its unique accepting run in $\cA$ is exactly negated in $\cA^-$. We therefore have the following.
\begin{proposition}
    \label{prop:negation of unambig}
    If $\cA$ is an unambiguous WFA, then $\cA^-$ is also an unambiguous WFA and for every word $w\in \Sigma^*$ we have either $\cA(w)=\cA^-(w)=\infty$ or $\cA(w)<\infty$ and $\cA^-(w)=-\cA(w)$.
\end{proposition}

Next, we recall the \emph{product construction} for WFAs, by which we can \emph{sum} two WFAs. Consider two WFAs $\cA_i=\tup{Q_i,\Sigma,q^i_0,\Delta_i,F_i}$ for $i\in \{1,2\}$. We define their product WFA $\cB=\tup{Q_1\times Q_2, \Sigma, (q^1_0,q^2_0),\Delta,F_1\times F_2}$ where 
\[
\Delta=\{((q_1,q_2),\sigma,c_1+c_2,(p_1,p_2))\mid (q_1,\sigma,c_1,p_1)\in \Delta_1 \wedge (q_2,\sigma,c_2,p_2)\in \Delta_2\}
\]
The following is folklore (see e.g.,~\cite[Section 5.1]{Almagor2020Whatsdecidableweighted})
\begin{proposition}
\label{prop:product construction}
    In the notations above, for every $w\in \Sigma^*$ we have $\cB(w)=\cA_1(w)+\cA_2(w)$.
\end{proposition}

\section{Proof of \cref{thm:unambig iff bounded U gap}}
\label{apx:thm unambig iff bounded U gap}
We can now present the detailed proof of \cref{thm:unambig iff bounded U gap}, based on the definition of gaps in \cref{def:U type B gap witness}.
\unambigchar*
We split the proof to two directions.
\subsection{$\cA$ is Unambiguisable $\implies$ Bounded Gaps}
\label{apx:sec:unambig implies bounded gap}
Let $\cA= \tup{Q,\Sigma, q_0, \Delta, F}$ and assume $\cA$ is unambiguisable. Let $\cU=\tup{S,\Sigma,s_0,\Lambda,G}$ be an equivalent unambiguous WFA. Combining \cref{prop:negation of unambig,prop:product construction}, we can obtain a WFA $\cB=\tup{Q\times S,\Sigma,(q_0,s_0),\Theta,F\times G}$ by negating $\cU$ and taking the product with $\cA$, so that for every $w\in \Sigma^*$ we have that either $\cA(w)=\cU(w)=\cB(w)=\infty$, or $\cB(w)=\cA(w)-\cU(w)=0$.

Let $\bigM=\max\{\norm{\cA},\norm{\cU},\norm{\cB}\}$ denote the maximal weight appearing in any of $\cA,\cU$ and $\cB$ in absolute value. Thus, in a single transition, any run of these WFAs can change the weight by at most $\bigM$.

Assume by way of contradiction that $\cA$ does not have bounded gaps. In particular, there exists a $B$-gap witness for $B>2|S||Q|\bigM+1$, given by $x,y\in \Sigma^*$, $p_1,q_1\in Q,p_2,q_2\in F$ and runs $\rho:q_0\runsto{x}p_1\runsto{y}p_2$ and $\chi:q_0\runsto{x}q_1\runsto{y}q_2$ as per \cref{def:U type B gap witness}.

Consider the single accepting run $\pi:s_0\runsto{x}s_1\runsto{y} s_2$ of $\cU$ on $xy$ ($\pi$ exists since $\cU$ accepts $xy$, and is unique since $\cU$ is unambiguous). 
We can now lift $\rho$ and $\chi$ to accepting runs of $\cB$ on $xy$ of the form 
$(\rho,\pi): (q_0,s_0)\runsto{x}(p_1,s_1)\runsto{y}(p_2,s_2)$ and $(\chi,\pi): (q_0,s_0)\runsto{x}(q_1,s_1)\runsto{y}(q_2,s_2)$.

By the construction of $\cB$ and using the gap property, we observe that 
\[
\begin{split}
&\cB((\rho,\pi)[x])-\cB((\chi,\pi)[x])\\
&={\cA}(\rho[x])-{\cU}(\pi[x])-({\cA}(\chi[x])-{\cU}(\pi[x]))\\
&={\cA}(\rho[x])-{\cA}(\chi[x])>B>2|S||Q|\bigM+1
\end{split}
\]
It follows that either $\cB((\rho,\pi)[x])>|S||Q|\bigM+1$ or $\cB((\chi,\pi)[x])<-|S||Q|\bigM-1$ (or both). We now split to these two cases, as depicted in \cref{apx:fig:unambiguisable to bounded gaps}. Intuitively, in the latter case we can find a shorter suffix that leads to a negative-weight run, which is a contradiction. In the former case we can find a negative cycle that can be pumped to a negative-weight run.

\begin{figure}[ht]
\centering
\begin{subfigure}[b]{0.45\textwidth}
\centering
\begin{tikzpicture}[
    >=stealth,
    thick,
    every node/.style={inner sep=1pt},
    dot/.style={circle,draw,inner sep=1pt},
    mindot/.style={circle,draw,double,inner sep=1pt}
]

  \node[dot]   (q0) at (0,0)   {$q_0$};
  \node[dot]   (q1) at (3.5,-1.4)   {$q_1$};
  \node[dot]   (p1) at (3.5,0.6) {$p_1$};
  \node[mindot] (p2) at (6,0)   {$p_2$};
  \node[mindot] (q2) at (6,0.7)   {$q_2$};
  \node[mindot] (q2b) at (5,-0.8)   {$q'_2$};

  \draw[->,dashed] (q0) -- node[above,pos=0.7] {$\chi$} (q1);
  \draw[->] (q0) -- node[above,pos=0.7] {$\rho$} (p1);
  \draw[dotted] (q0) -- node[above,pos=0.7] {$\chi$} (p2);

 \draw[<->,blue] (q1) -- node[fill=white,
    inner sep=2pt
] {$>\!B$} (p1);

  \draw[->] (p1) -- (p2);
  \draw[dashed,->] (q1) -- (q2);
  \draw[dashed,->] (q1) -- node[below,pos=0.7] {$y'$} (q2b);

  \draw[decorate,decoration={brace,mirror,amplitude=4pt}]
        (0,-1.8) -- node[below=4pt] {$x$} (3.5,-1.8);
  \draw[decorate,decoration={brace,mirror,amplitude=4pt}]
        (3.5,-1.8) -- node[below=4pt] {$y$} (6,-1.8);

  \node[right=10pt] at (p2) {$\min$};

\end{tikzpicture}
\caption{$\chi$ becomes too negative.}
\label{apx:fig:unambig to bounded gap negative case}
\end{subfigure}
\hfill
\begin{subfigure}[b]{0.47\textwidth}
\centering
\begin{tikzpicture}[
    >=stealth,
    thick,
    every node/.style={inner sep=1pt},
    dot/.style={circle,draw,inner sep=1pt},
    mindot/.style={circle,draw,double,inner sep=1pt}
]

  \node[dot]   (q0) at (0,0)   {$q_0$};
  \node[dot]   (q1) at (3.5,-0.6)   {$q_1$};
  \node[dot]   (p1) at (3.5,1.4) {$p_1$};
  \node[mindot] (p2) at (6,0)   {$p_2$};
  \node[mindot] (q2) at (6,0.7)   {$q_2$};
  \node[mindot] (q2b) at (7,-0.8)   {$q'_2$};

  \draw[->,dashed] (q0) -- node[above,pos=0.7] {$\chi$} (q1);
  \draw[->] (q0) -- node[above,pos=0.7] {$\rho$} (p1);
  \draw[dotted] (q0) -- node[above,pos=0.7] {$\chi$} (p2);

 \draw[<->,blue] (q1) -- node[fill=white,
    inner sep=2pt
] {$>\!B$} (p1);

  \draw[->] (p1) -- (p2);
  \draw[dashed,->] (q1) -- (q2);
  \draw[dashed,->] (p1) -- (q2b);

  \draw[decorate,decoration={brace,mirror,amplitude=4pt}]
        (0,-1) -- node[below=4pt] {$x$} (3.5,-1);
  \draw[decorate,decoration={brace,mirror,amplitude=4pt}]
        (3.5,-1) -- node[below=4pt] {$y$} (6,-1);

  \node[right=10pt] at (p2) {$\min$};

\end{tikzpicture}
\caption{$\rho$ decreases too much, leading to a negative cycle.}
\label{apx:fig:unambig to bounded gap positive case}
\end{subfigure}
\caption{Contradiction scenarios for \cref{apx:sec:unambig implies bounded gap} (the $s$ component is omitted). In \cref{apx:fig:unambig to bounded gap negative case} the run $\chi$ becomes too negative, so that a short suffix induces a negative run to $q'_2$. In \cref{apx:fig:unambig to bounded gap positive case}, the run $\rho$ decreases too much between $p_1$ and $p_2$, causing a negative cycle, which again leads to a negative run to $q'_2$.}
\label{apx:fig:unambiguisable to bounded gaps}
\end{figure}

 \paragraph*{If $\cB((\chi,\pi)[x])<-|S||Q|\bigM-1$} then since there is an accepting run $(q_1,s_1)\runsto{y}(q_2,s_2)$, there is also an accepting simple path from $(q_1,s_1)$ to $(q_2,s_2)$. 
Such a path induces a word $y'$ with $|y'|\le |S||Q|$ whose minimal-value accepting run $\tau:(q_1,s_1)\runsto{y'}(q_2,s_2)$ from $(q_1,s_1)$ in $\cB$ accumulates weight at most $|y'|\bigM \le |S||Q|\bigM$. It follows that $\cB((\chi,\pi)[x]\cdot \tau)<-|S||Q|\bigM-1+|S||Q|\bigM <0$. Therefore, $\cB(xy')<0$, in contradiction to the construction of $\cB$ (as it cannot assign non-zero weights).

 \paragraph*{If $\cB((\rho,\pi)[x])>|S||Q|\bigM+1$}
then recall that by the gap property, we have that $\rho$ is a minimal-weight run of $\cA$ on $xy$. Therefore, by the unambiguity of $\cU$, $(\rho,\pi)$ is also a minimal-weight run of $\cB$ on $xy$, and therefore $\cB((\rho,\pi))=0$. 
In particular, the suffix $(\rho,\pi)[y]:(p_1,s_1)\runsto{y}(p_2,s_2)$ has
$\cB((\rho,\pi)[y])<-|S||Q|\bigM-1$.
Therefore, there exists a negative-weight cycle along $(\rho,\pi)[y]$. Indeed, otherwise the minimal weight that can be accumulated is at most $-|S||Q|\bigM$. By repeating such a cycle we obtain a negative-weight run of $\cB$ to an accepting state (after concatenating it to $(\rho,\pi)[x]$). This is again a contradiction to the construction of $\cB$.

We conclude that $\cB$ has bounded gaps.

\subsection{Bounded Gaps $\implies$ $\cA$ is Unambiguisable}
\label{apx:sec:bounded gap implies unambig}
Assume that $\cA$ has gaps bounded by $B$. Before we construct an equivalent unambiguous WFA $\cU$, we define a notion of a \emph{canonical minimal run} of $\cA$ on a word $w$. This notion is then the crux of the construction: we build an unambiguous WFA that can track this canonical run. We illustrate this in \cref{xmp:unambiguisation and canonical}.

Fix some arbitrary linear order $\preceq$ on the states $Q$. We think of this order as a priority, where higher priority states are better. Consider a word $w=\sigma_1\cdots \sigma_n$ accepted by $\cA$ and let $\Upsilon$ be the set of minimal-weight accepting runs of $\cA$ on $w$. 
Since $\Upsilon$ is finite, we can denote its runs by $\{\rho^i=q_0^i,\ldots,q_n^i\mid 1\le i\le m\}$ for some $m$. 
We now describe a procedure for culling runs from $\Upsilon$ until we are left with a single run. 

Consider the sequence $\Upsilon_{n+1}\supset \Upsilon_n\supset\ldots \supset \Upsilon_0$ defined inductively (from $n$ to $0$) as follows.
\begin{itemize} 
    \item $\Upsilon_{n+1}=\Upsilon$.
    \item For $0\le k\le n$ we define $\Upsilon_{k}=\{\rho^i\mid \rho^i\in \Upsilon_{k+1}\wedge q_{k}^{i'}\preceq q_k^{i} \text{ for every }i' \text{ such that }\rho^{i'}\in \Upsilon_{k+1}\}$.
\end{itemize}
Intuitively, we consider the set of all minimal runs on $w$, and start scanning them from the end backwards. We first remove all runs for which $q^i_n$ is not $\preceq$-maximal. Then, from the remaining runs (if there are more than one), we keep only runs where $q^i_{n-1}$ is $\preceq$-maximal, and so on.

Note that for $0\le k\le n$, the runs in  $\Upsilon_k$ are all identical from index $k$. Therefore, $\Upsilon_0$ has a single run $\canrho$, which we dub the \emph{canonical run on $w$}.
By definition, $\canrho$ is a minimal run of $\cA$ on $w$. Also, since $\preceq$ is a linear order, the procedure above is deterministic, meaning that $\canrho$ is uniquely defined given $\preceq$.

We can now proceed to construct an equivalent unambiguous WFA $\cU$. We start with a brief intuition.
Upon reading a word $w$, the WFA $\cU$ attempts to track the canonical minimal run of $\cA$ on $w$. To do so, $\cU$ keeps track of all the runs in a window of weight $\pm B$ around a (nondeterministically chosen) state $q$. If all the runs stay close to $q$, then all the runs are tracked. However, once a run becomes too high or too low, the window tracks it as $\infty$ or $-\infty$, respectively.
Then, when the word ends, if the current state $q$ is accepting, has minimal weight in the window (in particular there are no accepting runs with weight $-\infty$) and $q$ has maximal priority, then this state accepts.

The main idea is that due to the gap property, if we indeed track the canonical run, then all other accepting states end within its $\pm B$ window, with higher weight or lower priority. In addition, other accepting runs that do not become minimal do not yield accepting runs of $\cU$, since their windows invariably ``believe'' that the canonical run has lower weight or higher priority, and therefore are not marked as accepting.

We now turn to the precise construction. A \emph{$B$-window} is a function $f:Q\to \{-\infty,-B,\ldots, B,\infty\}$, and we denote the set of such functions as $\Bwindows=\{-\infty,-B,\ldots, B,\infty\}^Q$. Intuitively, assume we are tracking a certain state $q$ with weight $0$. A $B$-window $f$ ``around $q$'' prescribes for each state $p$ whether the minimal weight with which $p$ can be reached is within distance $B$ from $0$, or whether it is more than $B$ above or below ($\infty$ and $-\infty$, respectively). Note that this is only intuition, and the precise details contradict it at certain points (which we mention later on).
In the following, we assume some arbitrary linear order $\preceq$ on the states $Q$.

We define $\cU=\tup{S,\Sigma,s_0,\Lambda,G}$ with the following components.
\begin{itemize} 
    \item The states are $S=Q\times \Bwindows$. Intuitively, each state tracks a state $q\in Q$ and a $B$-window around $q$.
    \item The initial state is $(q_0,f_0)$ where $f_0(q_0)=0$ and $f_0(p)=\infty$ for all $p\neq q_0$.
    \item The accepting states are
    \[G=\{(q,f_q)\mid q\in F \wedge \forall p\in F, (f_q(p)>0 \vee (f_q(p)=0 \wedge p\preceq q))\}\]
    That is, a state $(q,f_q)$ is accepting if $q\in F$ and $q$ has minimal weight among the accepting state in the window (and this weight is $0$). In case there are several accepting states with weight $0$, the state is accepting if $q$ is the maximal among them in the state ordering $\preceq$. 
    \item The transition relation $\Lambda$ is defined as follows. Consider $(q,f_q)\in S$ and a letter $\sigma\in \Sigma$. For each transition $(q,\sigma,c,p)\in \Delta$ we may introduce a transition in $\Lambda$, according to the following procedure.
    \begin{itemize}
        \item  We first update $f_q$ under this transition, by constructing an intermediate function $g:Q\to \bbZ\cup\{-\infty,\infty\}$ where for $p\in Q$ we define 
    \[g(p)=\min\{f_{q}(r)+\minweight(r\runsto{\sigma}p)-c\mid r\in Q\}
    \]
    Note that the latter can be $\infty$ if $p$ is reachable only from states with weight $\infty$ in $f_q$, and can be $-\infty$ if $p$ is reachable from some state with weight $-\infty$ in $f_q$.
    Also note that we normalise the weight of a transition by $-c$ (where $c$ is the weight of the transition we focus on).
    \item We now perform two ``consistency'' checks on $g$.
    \begin{itemize}
        \item If $g(p)<0$, then we do not introduce a transition. Intuitively, if $p$ can be reached with lower-weight, then another run of $\cU$ already tracks the lower-weight run.
        \item If there exists $r\neq q$ with $q\preceq r$ such that $f_{q}(r)+\minweight(r\runsto{\sigma}p)-c=g(p)$, then we do not introduce a transition. Intuitively, if there is a higher-priority state $r$ that yields $p$ with the same weight, then another run of $\cU$ already tracks this higher-priority run (c.f., tracking the canonical run).

    \end{itemize}
    \item If the check above succeeds, we turn $g$ into a $B$-window by capping its entries at $-B$ and $B$. That is, for every $r\in Q$ let $f_p(r)=\infty$ if $g(r)>B$ and $f_p(r)=-\infty$ if $g(r)<-B$, and otherwise let $f_p(r)=g(r)$.
    \item  We then add the transition $((q,f_q),\sigma,c,(p,f_p))$ to $\Lambda$. We say that this transition is \emph{lifted from $(q,\sigma,c,p)$}.
    \end{itemize}
\end{itemize}
It remains to prove that $\cU$ is equivalent to $\cA$ and that $\cU$ is unambiguous.
Before proceeding, we present two key lemmas regarding the behaviour of $\cU$. 
Intuitively, \cref{lem:characterization of unambig equivalent} shows how the $f_q$ component of the run tracks the minimal runs around $q$, assuming a given run of $\cU$.
In \cref{lem:canonical run lifts} we show that the canonical run on a word lifts to a run of $\cU$.

\begin{lemma}
\label{lem:characterization of unambig equivalent}
Consider a word $w=\sigma_1\cdots\sigma_n$ and a run $\pi=(q_0,f_0),(q_1,f_1),\ldots, (q_n,f_n)$ of $\cU$ on $w$. Let $\rho=q_0,\ldots,q_n$ be the run of $\cA$ from which $\pi$ lifts.
That is, for every $0\le i<n$, let $(q_i,\sigma_{i+1},c_{i+1},q_{i+1})\in\Delta$ be the corresponding transition in $\rho$, so that in $\pi$ the corresponding transition is
$((q_i,f_i),\sigma_{i+1},c_{i+1},(q_{i+1},f_{i+1}))\in\Lambda$.

For every $0\le i\le n$ and $p\in Q$, if $f_i(p)\neq \infty$, define
\[
\offset_i(p)=\minweight_\cA(q_0\runsto{w[1,i]}p)-\weight(\rho[1,i])\in\bbZ\cup\{\pm\infty\}.
\]
Then for every $0\le i\le n$ and $p\in Q$, the following hold:
\begin{enumerate} 
    \item $f_i(q_i)=0$. That is, the weight assigned to the current state being tracked remains at $0$.
    
    \item $f_i(p)\in\{-B,\ldots,B\}$ if and only if there is a minimal-weight run $\tau:q_0\runsto{w[1,i]}p$, $\tau=s_0,s_1,\ldots, s_i$ such that
    $|\offset_j(s_j)|\le B$ for all $0\le j\le i$; in that case,
    $f_i(p)=\offset_i(p)$.
    
    \item $f_i(p)=-\infty$ if and only if there exists a run $\tau:q_0\runsto{w[1,i]}p$, $\tau=s_0,s_1,\ldots, s_i$ such that the minimal index
    \[i_0=\min\{j\mid j\le i \wedge |\offset_j(s_j)|>B\}\] 
    exists (i.e., the minimum is not empty) and $\offset_{i_0}(s_{i_0})<-B$.
    
    \item $f_i(p)=\infty$ if and only if every run $\tau:q_0\runsto{w[1,i]}p$, $\tau=s_0,s_1,\ldots, s_i$ satisfies that the minimal index
    \[i_0=\min\{j\mid j\le i \wedge |\offset_j(s_j)|>B\}\] 
    exists (i.e., the minimum is not empty) and $\offset_{i_0}(s_{i_0})>B$.
\end{enumerate}
\end{lemma}

\begin{proof}
The proof follows from the definition of $\Lambda$ by induction on $i$. 
 \subparagraph*{Base case ($i=0$).}
We have $w[1,0]=\epsilon$ and $\weight(\rho[1,0])=0$. By definition,
$f_0(q_0)=0$ and $f_0(p)=\infty$ for all $p\neq q_0$.
Also $\offset_0(q_0)=0$ and $\offset_0(p)=\infty$ for $p\neq q_0$,
so the three statements hold. 

 \subparagraph*{Induction step.}
Assume the claim holds for index $i$.
Fix some $p\in Q$, and recall that $f_{i+1}(p)$ is determined by capping \[g(p)=\min\{f_{q}(r)+\minweight(r\runsto{\sigma}p)-c\mid r\in Q\} \]
\begin{enumerate} 
    \item $f_{i+1}(q_{i+1})=0$ follows directly by the definition of $\Lambda$ (since $\minweight(q_i\runsto{\sigma_{i+1}}q_{i+1})=c_{i+1}$, and since the first check on $g$ passes).
    \item By the definition of $\Lambda$, we have that $f_{i+1}(p)\in \{-B,\ldots,B\}$ if and only if the minimum in $g_{i+1}(p)$ above is also between $\{-B,\ldots,B\}$. 
    By the induction hypothesis, the state $r$ for which this minimum is attained satisfies $f_i(r)\in \{-B,\ldots,B\}$. We can then use the induction hypothesis to obtain a run to $r$ whose prefixes satisfy the offset criterion. 
    Composing those with the transition $r\runsto{\sigma_{i+1}}p$ yields a minimal weight run $\tau$ as required, and in particular $|\offset_{i+1}(p)|\le B$. 
    Conversely, the existence of such a minimal run again implies by the induction hypothesis that $f_{i+1}(p)\in \{-B,\ldots, B\}$.
    \item Similarly, the minimum above is $-\infty$ if there is some $r\in Q$ with $r\runsto{\sigma_{i+1}}p$ such that either $f_{i}(r)=-\infty$ (in which case Item 3 follows by induction), or $f_i(r)+\minweight(r\runsto{\sigma_{i+1}} p)-c_{i+1}<-B$, which is equivalent (by the induction hypothesis) to the existence of a run $\tau$ as required.
    \item Finally, the minimum above is $\infty$ if every $r\in Q$ with $r\runsto{\sigma_{i+1}}p$ satisfies that either $f_i(r)=\infty$ or $f_i(r)+\minweight(r\runsto{\sigma_{i+1}}p)-c_{i+1}>B$. Again, by the induction hypothesis this is equivalent to all runs $\tau:q_0\runsto{w[1,i+1]}p$ satisfying the required condition.
\end{enumerate}
\end{proof}
In \cref{lem:characterization of unambig equivalent} we assume that we start with some existing run of $\cU$. However, the consistency checks on $g$ are not simple to meet, and it is not clear that such runs exist. The following lemma shows that the canonical run can be lifted to a run of $\cU$.

\begin{lemma}
\label{lem:canonical run lifts}
Consider a word $w=\sigma_1\cdots\sigma_n$ and the canonical run $\canrho:q_0\runsto{w}q_n$ of $\cA$ on $w$ denoted $\canrho=q_0,q_1,\ldots,q_n$.
Let $\pi=(q_0,f_0),(q_1,f_1),\ldots, (q_n,f_n)$ be the sequence of lifted transitions of $\cU$ on $w$ induced by $\canrho$. That is, for every $0\le i<n$, let $(q_i,\sigma_{i+1},c_{i+1},q_{i+1})\in\Delta$ be the corresponding transition in $\canrho$, then we take in $\pi$ the transition
$((q_i,f_i),\sigma_{i+1},c_{i+1},(q_{i+1},f_{i+1}))\in\Lambda$.
Then $\pi$ is a run of $\cU$. 
\end{lemma}

\begin{proof}
The proof follows from the definition of $\Lambda$ by induction on $i$. At every step we show that the transition exists in $\cU$, i.e., that it passes the checks on $g$ in the definition of $\Lambda$.

 \subparagraph*{Base case ($i=0$).}
The initial state in $\pi$ is $(q_0,f_0)$, which is the initial state of $\cU$.

 \subparagraph*{Induction step.}
Assume the claim holds for index $i$.
We show that the transition 
\[((q_{i},f_i),\sigma_{i+1},c_{i+1},(q_{i+1},f_{i+1}))\] 
is in $\Lambda$.
By the definition of $\Lambda$, $f_{i+1}$ is obtained from the intermediate function
\[
g_{i+1}(p)=
\min\{ f_i(r)+\minweight(r\runsto{\sigma_{i+1}}p)-c_{i+1} | r\in Q \}
\]
by applying the capping rule as per the definition of $\cU$, if $g$ passes the two consistency checks. We now verify that both checks pass. The proof relies on \cref{lem:characterization of unambig equivalent}, and in particular uses the $\offset$ notation with respect to the run $\canrho$.
\begin{itemize}
    \item Assume by way of contradiction that $g_{i+1}(q_{i+1})<0$, then there exists some $r\in Q$ such that $f_i(r)+\minweight(r\runsto{\sigma_{i+1}}q_{i+1})-c_{i+1}<0$. 
    We split to two cases.
    If $f_i(r)\in \bbZ$ then intuitively we found a lower-weight run than $\canrho$, which is a contradiction. Formally, recall that $c_{i+1}=\minweight(q_i\runsto{\sigma_{i+1}}q_{i+1})$, then reordering gives us
    \[
        \minweight(r\runsto{\sigma_{i+1}}q_{i+1})<\minweight(q_i\runsto{\sigma_{i+1}}q_{i+1})-f_i(r)
    \]
    Since the run $\pi[1,i]$ is valid by the induction hypothesis, then \cref{lem:characterization of unambig equivalent} gives us
    \[f_i(r)=\offset_i(r)=\minweight(q_0\runsto{w[1,i]}r)-\weight(\canrho[1,i])\]
    We can now use the two equations above to get
    \[
    \begin{split}
    &\minweight(q_0\runsto{w[1,i+1]}q_{i+1})\le \minweight(q_0\runsto{w[1,i]}r)+\minweight(r\runsto{\sigma_{i+1}}q_{i+1})<\\
    &f_i(r)+\weight(\canrho[1,i])+\minweight(q_i\runsto{\sigma_{i+1}}q_{i+1})-f_i(r)=\weight(\canrho[1,i+1])
    \end{split}
    \]
    This, however, contradicts the fact that $\canrho$ is a minimal-weight run.

    The second case is when $f_i(r)=-\infty$. Intuitively, in this case there is an accepting run that at some point went outside the $B$ window of $\canrho$, in contradiction to the bounded gap property. 
    Formally, again by the induction hypothesis and \cref{lem:characterization of unambig equivalent} there is a run $\tau:q_0\runsto{w[1,i]}r$ that at some index $i_0$ gets offset less than $-B$. 
    We then have $\weight(\canrho[1,i_0])-\weight(\tau[1,i_0])>B$. However, this is a contradiction to the bounded gaps of $\cA$ (with $x=w[1,i_0-1]$ and $y=w[i_0,n]$).

    We remark that despite the seeming simplicity of the second case, it hides an intriguing behaviour: it may be that the actual run $\tau$, after being tracked as $-\infty$, actually rises in weight and becomes higher than $\canrho$. In such a setting, tracking $\tau$ as $-\infty$ is actually contradicting to the intuition of the construction. Technically, however, it plays no role.

    We conclude that the first consistency check of $g_{i+1}$ passes.

    \item For the second consistency check, assume by way of contradiction that there exists $r\neq q_{i}$ with $q_i\preceq r$ such that $f_i(r)+\minweight(r\runsto{\sigma_{i+1}}q_{i+1})-c_{i+1}=g_{i+1}(q_{i+1})$.
    In particular, we have that $f_i(r)\in \bbZ$, and by the induction hypothesis and \cref{lem:characterization of unambig equivalent} there exists a run $\tau:q_0\runsto{w[1,i]}r$ such that $\weight(\tau)=\weight(\canrho[1,i])+f_i(r)$. Consider the run $\tau\cdot\canrho[i+1,n]$. By the above we have $\weight(\tau\cdot\canrho[i+1,n])=\weight(\canrho)$. Moreover, $\canrho$ and $\tau\cdot\canrho[i+1,n]$ are identical in their suffix from $i+1$. 
    However, since $q_i\preceq r$ and $q_i\neq r$, it follows that $\canrho$ would be culled in $\Upsilon_{i}$, in contradiction to it being the canonical run.

    It follows that the second check also passes.
\end{itemize}
Since both checks pass, we conclude that the transition $((q_{i},f_i),\sigma_{i+1},c_{i+1},(q_{i+1},f_{i+1}))$ exists in $\Lambda$, and we are done.
\end{proof}

 \paragraph*{$\cU$ is equivalent to $\cA$} 
We first show that for every word $w$ we have ${\cA}(w)\le \cU(w)$. 
Consider a word $w=\sigma_1\cdots\sigma_n$ that is accepted by $\cU$ and let $\pi:q_0\runsto{w}F$ be a minimal-weight accepting run of $\cU$ on $w$. 
Write $\pi=(q_0,f_0),(q_1,f_1),\ldots, (q_n,f_n)$. 
We claim that the projection of $\pi$ onto $\cA$ is an accepting run of $\cA$ on $w$, with the same weight as that of $\pi$. More precisely, let 
$\rho=q_0,q_1,\ldots,q_n$ be the corresponding run of $\cA$ on $w$. Clearly this is a legal and accepting run, since each transition $((q_i,f_i),\sigma_{i+1},c_{i+1},(q_{i+1},f_{i+1}))\in \Lambda$ in $\pi$ is lifted from a transition $(q_i,\sigma_{i+1},c_{i+1},q_{i+1})\in \Delta$, and since $q_n\in F$ (by the definition of $G$).
Moreover, the weight of the transitions remains $c_{i+1}$, meaning that the run accumulates the same weight as $\pi$. It follows that 
\[\cA(w)=\minweight_{\cA}(q_0\runsto{w}F)\le {\cA}(\rho)={\cU}(\pi)={\cU}(w)\]
Note that this direction relies only on the fact that $\cU$ tracks the runs of $\cA$ in its first component, and does not rely on the gap property, nor on the special structure of $\cU$.

We now turn to show that for every word $w$ we have ${\cU}(w)\le \cA(w)$. 
Consider a word $w=\sigma_1\cdots \sigma_n$, and let $\canrho=q_0,\ldots,q_n$ be the canonical minimal-weight accepting run of $\cA$ on $w$. 
By \cref{lem:characterization of unambig equivalent} the run $\canrho$ lifts to a run $\pi$ of $\cU$ on $w$. Moreover, the sequence of weights accumulated by this run is identical to that of $\canrho$, so $\cA(\canrho)=\cU(\pi)$. 
It remains to show that $\pi$ is accepting. Still by \cref{lem:characterization of unambig equivalent} we have that the last state $(q_n,f_n)$ in $\pi$ satisfies $f_n(q_n)=0$, and $q_n\in F$ since $\canrho$ is accepting. 
Following the definition of the accepting states $G$, assume by way of contradiction that there exists an accepting state $p\in F$ such that $f_n(p)<0$ or $f_n(p)=0$ and $q_n\preceq p$. The former case implies the existence of an accepting run with lower weight than $\canrho$, in contradiction to the minimality of $\canrho$. 
The latter implies that $\canrho$ is culled at $\Upsilon_n$, in contradiction to $\canrho$ being the canonical run.
We conclude that $(q_n,f_n)$ is accepting. Thus, \[\cU(w)=\minweight_{\cU}((q_0,f_0)\runsto{w}G)\le \cU(\pi)= {\cA}(\canrho)={\cA}(w)\]

 \paragraph*{$\cU$ is unambiguous}
It remains to prove that $\cU$ is unambiguous. Consider a word $w=\sigma_1\cdots\sigma_n$ that is accepted by $\cU$. Thus, $w$ is also accepted by $\cA$, and therefore the canonical run $\canrho:q_0\runsto{w}q_n$ is defined.
Let $\pi=(s_0,f_0),\ldots,(s_n,f_n)$ be an accepting run of $\cU$ on $w$. We claim that $\pi$ is the run lifted from $\canrho$, and is therefore unique.

Let $\tau=s_0,\ldots,s_n$ be the state-projection of the run $\pi$. By the definition of $\Lambda$ this is an accepting run of $\cA$ on $w$.
We prove by reverse induction, from $n$ to $0$, that $s_i=q_i$ for all $i$. In particular, this means that $\tau=\canrho$, and therefore $\pi$ is lifted from $\canrho$.

For the base case, we have that $(s_n,f_n)\in G$. By \cref{lem:characterization of unambig equivalent} (Item 1) we have $f_n(s_n)=0$. Since $\canrho$ is a minimal-weight run of $\cA$, then $\weight(q_0\runsto{w}q_n)\le \weight(q_0\runsto{w}s_n)$. By \cref{lem:characterization of unambig equivalent} (Items 2,3) we have $f_n(q_n)\le 0$. If $f_n(q_n)<0$, then since $q_n\in F$ this violates the condition in $G$, which is a contradiction. Thus, we have $f_n(q_n)=0$. In particular, by \cref{lem:characterization of unambig equivalent} this means that $\tau$ is also a minimal accepting run of $\cA$ (since $\canrho$ has offset $0$ from it).
Now, since $(s_n,f_n)\in G$, then $q_n\preceq s_n$. On the other hand, $q_n$ is the $\preceq$-maximal state among all minimal-accepting runs, by the culling process in $\Upsilon_n$, so $s_n\preceq q_n$. It follows that $s_n=q_n$.

We proceed to the induction case. Assume $s_{i+1},\ldots,s_n=q_{i+1},\ldots,q_n$. We prove that $s_i=q_i$. 
Since $\canrho$ is a minimal-weight run of $\cA$, then by the bounded gap property we have that $|\weight(\canrho[1,j])-\weight(\tau[1,j])|\le B$ for all $1\le j\le n$, and in particular for $1\le j\le i$. By \cref{lem:characterization of unambig equivalent} (Item 2) this means that $f_i(q_i)\in \{-B,\ldots,B\}$. 
By the induction hypothesis, since $q_{i+1}=s_{i+1}$, we have $f_{i+1}(q_{i+1})=f_{i+1}(s_{i+1})=0$. 
By the definition of $\Lambda$, this means that
$f_i(q_i)+\minweight(q_i\runsto{\sigma_{i+1}}q_{i+1})-c_{i+1}=0=f_{i+1}(s_{i+1})$.
By the second consistency check, this implies that $s_{i}\preceq q_i$ (otherwise the test would fail).
However, note that $s_0,\ldots,s_i,s_{i+1},\ldots,s_n$ and $q_0,\ldots,q_i,q_{i+1}\ldots,q_n$ share the $[i+1,n]$ suffix and are both accepting runs. Therefore, by the culling process at $\Upsilon_i$ we have that $s_i\preceq q_i$. Thus, we conclude that $s_i=q_i$ and we are done.

We conclude that $\tau=\canrho$, implying the uniqueness of the accepting run. Thus, $\cU$ is unambiguous.

\section{Correctness of the Reduction from Unambiguisability to Determinisability}
\label{apx:reduction correctness}
Before proceeding to prove the correctness of the construction, we establish a basic correspondence between $\cA$ and $\cB$. For a word $w\in \Gamma^*$, we denote by $w|_\Sigma$ its projection on $\Sigma^*$. Similarly, for a run $\rho$ of $\cB$ we denote by $\rho|_Q$ its projection on $Q$.
\begin{proposition}
    \label{prop: redcution correspondence}
    Consider a word $w=(\sigma_1,\alpha_1)\cdots(\sigma_n,\alpha_n)\in \Gamma^*$ and a run $\rho:s_0\runsto{w}(p,f)$ of $\cB$ on $w$, then $\rho|_Q:q_0\runsto{w|_\Sigma}p$ and for every $0\le i\le n$ we have $\weight_{\cB}(\rho[1,i])=\weight_{\cA}(\rho|_Q[1,i])$.

    Conversely, for every word $w'=\sigma_1\cdots \sigma_n\in \Sigma^*$ there is an \emph{update track} $u=\alpha_1\cdots \alpha_n\in \update^*$ and a \emph{commitment track} $\Theta=f_0,\ldots,f_n\in \commit^*$ such that for every run $\chi:q_0\runsto{w'}q_n$ of $\cA$ with $\chi=q_0,q_1,\ldots,q_n$, the sequence $\tau=(q_0,f_0),\ldots,(q_n,f_n)$ is a run of $\cB$ $\tau:(q_0,f_0)\runsto{(\sigma_1,\alpha_1)\cdots (\sigma_n,\alpha_n)}(q_n,f_n)$.
    Moreover, if $\chi$ is accepting then $\tau$ is accepting.

\end{proposition}
\begin{proof}
    For the first part, note that the first component of the states of $\cB$ follows exactly the transitions of $\cA$. In particular, the projection $\rho|_Q$ trivially satisfies the requirement.

    The second part is very slightly more nuanced. Given $w'=\sigma_1\cdots \sigma_n\in \Sigma^*$, we construct the update track and commitment track to follow the run structure of $\cA$ on $w'$. Starting with $f_0$ such that $s_0=(q_0,f_0)$ as per the definition of $\cB$, we define inductively for $0<i\le n$ the update $\alpha_{i}$ and commitment $f_{i}$ by 
    \[
    \alpha_i(p,q)=\begin{cases}
        \keep & p\runsto{\sigma_i}q\runsto{\sigma_{i+1}\cdots \sigma_n}q_\fin\\
        \trim & p\runsto{\sigma_i}q \wedge \minweight(q\runsto{\sigma_{i+1}\cdots \sigma_n}q_\fin)=\infty\\
        \bot & (p,\sigma_i,\infty,q)\in \Delta
    \end{cases}
\]
\[f_i(p)=\begin{cases}
  \keep &
    (\exists r\in Q,\ f_{i-1}(r)=\keep \wedge\ r\keep p\in\alpha_i) \wedge (\forall r\in Q,\ f_{i-1}(r)\neq \bot \implies r\trim p \notin \alpha_i)\\[0pt]
     \trim & \exists r\in Q,\ f_{i-1}(r)\neq \bot \wedge r\trim p\in \alpha_i 
     \\[0pt]
    \bot &
    \forall r\in Q,\ f_{i-1}(r)\neq \bot \implies r\bot p\in \alpha_i
\end{cases}
    \]
Intuitively, $f_i(p)=\keep$ if it is $\keep$-reached from some $\keep$ state, and no other reachable state marks it as $\trim$, and $f_i(p)=\trim$ if some reachable state marks it as $\trim$. Otherwise it is unreachable.
    
    Now, let $\chi:q_0\runsto{w'}q_n$ be a run, we construct a run $\tau=(q_0,f_0),\ldots,(q_n,f_n)$ of $\cB$ on $w=(\sigma_1,\alpha_1)\cdots (\sigma_n,\alpha_n)$. 
    By the definition of $\Lambda$, it readily follows that the transitions pass all consistency checks. 
    Moreover, if $q_n=q_\fin$, it follows that $f_n(q_\fin)=\keep$, and by the construction of $f_n$ we immediately have $f_n(p)\neq \keep$ for every $p\neq q_\fin$. Thus, $(q_\fin,f_n)\in G$, so $\tau$ is accepting.
\end{proof}

\subsection{$\cB$ is Determinisable $\implies$ $\cA$ is Unambiguisable}
\label{sec: B is det to A is unambig}
We prove the contrapositive of this claim -- if $\cA$ is not \unable then $\cB$ is not determinisable. This is done via the gap characterisation. Concretely, we prove the following.
\begin{lemma}
\label{lem:Utype in A to Dtype in B}
    If there exists a \Utype $B$-gap witness in $\cA$, then there exists a \Dtype $B$-gap witness in $\cB$.
\end{lemma}
\begin{proof}
    Recall that $F=\{q_\fin\}$ and consider a \Utype $B$-gap witness $xy$ in $\cA$, then by \cref{def:U type B gap witness} there are $p_1,q_1\in Q$ such that there exist runs $\rho:q_0\runsto{x}p_1\runsto{y}q_\fin$ and $\chi:q_0\runsto{x}q_1\runsto{y}q_\fin$ where $\minweight(q_0\runsto{x}Q)=\weight(\chi[x])$, $\minweight(q_0\runsto{xy}q_\fin)=\weight(\rho)$, and 
    $\weight(\rho[x])-\weight(\chi[x])> B$.

    Denote $x=x_1\cdots x_n$ and $y=y_1\cdots y_m$. 
    By the second part of \cref{prop: redcution correspondence} we can construct update tracks $\alpha_1\cdots \alpha_n\in \update^*, \beta_1\cdots \beta_m\in \update^*$ and commitment tracks $f_0,\ldots, f_n,g_1,\ldots,g_m$ so that for the words $x'=(x_1,\alpha_1)\cdots (x_n,\alpha_n)$ and $y'=(y_1,\beta_1)\cdots (y_m,\beta_m)$ we have runs 
    $\rho':(q_0,f_0)\runsto{x'}(p_1,f_n)\runsto{y'}(q_\fin,g_m)$ and $\chi':(q_0,f_0)\runsto{x'}(q_1,f_n)\runsto{y'}(q_\fin,g_m)$. 

    We claim that $x'y'$ is a \Dtype $B$-gap witness in $\cB$.    
    By the first part of \cref{prop: redcution correspondence} we have that the weights of $\rho$ and $\rho'$ coincide for every prefix, and the same for $\chi$ and $\chi'$. In particular, it holds that $\weight(\rho'[x'])-\weight(\chi'[x'])>B$, satisfying the 3rd requirement of \cref{def: det B gap witness}. We continue to show the other two requirements.

    First, we claim that $\minweight_\cB((q_0,f_0)\runsto{x'}S)=\cB(\chi'[x'])$. Indeed, if by way of contradiction there exists a run $\tau:(q_0,f_0)\runsto{x'}S$ with $\cB(\tau)<\cB(\chi'[x'])$, then by the first part of \cref{prop: redcution correspondence} we have $\cA(\tau|_Q)<\cA(\chi[x])$, contradicting the first requirement of \cref{def:U type B gap witness}.
    
    An analogous argument shows that $\minweight_\cB((q_0,f_0)\runsto{x'y'}G)=\weight(\rho'[x'y'])$, concluding that $x'y'$ is indeed a \Dtype $B$-gap witness in $\cB$. 
\end{proof}

Using the lemma, we now have that if $\cA$ is \emph{not} \unable, then by \cref{thm:unambig iff bounded U gap} for every $B$ there is a \Utype $B$-gap witness in $\cA$. By \cref{lem:Utype in A to Dtype in B} there is also a \Dtype $B$-gap witness in $\cB$, and therefore by \cref{thm:det iff bounded gap} we have that $\cB$ is not determinisable. By the contrapositive, if $\cB$ is determinisable, then $\cA$ is \unable.

\subsection{$\cA$ is Unambiguisable $\implies$ $\cB$ is Determinisable}
\label{sec: A is unambig to B is det}
We proceed to the converse (and harder) correctness proof, where we actually use the properties in the construction of $\cB$. We present the converse of \cref{lem:Utype in A to Dtype in B}. Before proceeding, we assume without loss of generality that $\cB$ is \emph{trim}, i.e., that every state is reachable from the initial state, and can reach the accepting states. Trivially, every WFA is determinisable if and only if its trimmed version is determinisable. Note, however, that we need to use \cref{prop: redcution correspondence} carefully, so as not to induce runs to states that are trimmed from $\cB$. We comment on this when relevant.
\begin{lemma}
\label{lem:Dtype in B to Utype in A}
    If there exists a \Dtype $B$-gap witness in $\cB$, then there exists a \Utype $B$-gap witness in $\cA$.
\end{lemma}
\begin{proof}
    Consider a \Dtype $B$-gap witness $xy$ in $\cB$, then by \cref{def: det B gap witness} there are $(q,f_q),(p,f_p)\in S$, $(q_\fin,g)\in F$ and runs $\rho:(q_0,f_0)\runsto{x}(p,f_p)\runsto{y}(q_\fin,g)$ 
    and $\chi:(q_0,f_0)\runsto{x}(q,f_q)$ such that the following holds.
    \begin{itemize} 
        \item $\minweight(q_0\runsto{x}Q)=\weight(\chi)$, i.e. $\chi:q_0\runsto{x}q$ is a minimal-weight run on $x$ (not necessarily accepting).
        \item $\minweight(q_0\runsto{xy}F)=\weight(\rho)$, i.e., $\rho$ is a minimal accepting run on $xy$.
        \item $\weight(\rho[x])-\weight(\chi[x])> B$, i.e., after reading $x$ and reaching $q$, the run $\rho$ is at least $B$ above the minimal run $\chi$.        
    \end{itemize}
    By the first part of \cref{prop: redcution correspondence} (whose application is unrelated to $\cB$ being trim), this readily implies that the word $(xy)|_{\Sigma}$ and the runs $\rho|_Q$ and $\chi|_Q$ almost satisfy the conditions of being a \Utype witness in $\cA$. All that remains is to show that $\chi|_Q$ can be extended to an accepting run on $y|_\Sigma$. 
    That is, it suffices to prove that $q\runsto{y}q_\fin$ (in $\cA$), which we now turn to show.

    Recall that $\cB$ is assumed to be trim,
    and consider the state $(q,f_q)\in S$. We claim that $f_q(q)=\keep$. Indeed, since $\chi|_Q:q_0\to q$, then $q$ is reachable from $q_0$ so $f_q(q)\neq \bot$. If, by way of contradiction $f_q(q)=\trim$, then by the outgoing consistency (and by induction) for every  state $(q',f')$ and word $z$ such that $(q,f_q)\runsto{z}(q',f')$ it holds that $f'(q')\neq \keep$. In particular, $(q',f')\notin G$. But then $(q,f_q)$ cannot reach $G$, and would therefore be trimmed from $\cB$. Therefore, we have that $f_q(q)=\keep$. 

    Next, recall that in the construction of $\cB$, the commitment components are updated \emph{deterministically} according to the updates. Thus, we can in fact assume $f_p=f_q$ and in particular $f_p(p)=f_p(q)=\keep$.
    Intuitively, this means that the DAG of runs represented by the updates in $y$ ``commits'' to both $p\runsto{y|_\Sigma}q_\fin$ and $q\runsto{y|_\Sigma}q_\fin$. 
    Formally, denote $y=(y_1,\alpha_1)\cdots (y_n,\alpha_n)$ and $\rho[y]=(p_0,f_0),\ldots,(p_n,f_n)$ (where $(p_0,f_0)=(p,f_p)$ and $(p_n,f_n)=(q_{\fin},g)$). Since $f_n\in G$ we have $f_n(q_\fin)=\keep$ and $f_n(q')\neq \keep$ for all $q'\neq q_\fin$. By induction from $n$ to $0$, and using the incoming consistency, we observe that for every $0\le i\le n$, if $f_i(p')=\keep$ for some $p'$, then $p'\runsto{y_{i+1}\cdots y_n} q_\fin$ (for $i=n$, the base case, this run is over the empty word).

    Therefore, we indeed have that $q\runsto{y|_\Sigma}q_\fin$, so $\chi|_Q$ has an accepting continuation on $y|_\Sigma$. This concludes all the requirements for $(xy)|_\Sigma$ to be a \Utype $B$-gap witness.
\end{proof}

Using the lemma, we now have that if $\cB$ is \emph{not} determinisable, then so is the trimmed version of it. Then, by \cref{thm:det iff bounded gap} for every $B$ there is a \Dtype $B$-gap witness in $\cB$. By \cref{lem:Dtype in B to Utype in A} there is also a \Utype $B$-gap witness in $\cA$, and therefore by \cref{thm:unambig iff bounded U gap} we have that $\cA$ is not \unable. By the contrapositive, if $\cA$ is \unable, then $\cB$ is determinisable.

This concludes the proof of \cref{thm:main reduction}. Since determinisability is decidable by~\cite{almagor2026determinization}, we have the following.
\unambigdecidable*

\section{Cost Register Automata and Width-$k$ WFAs}
\label{apx:CRA}
 \paragraph*{Cost Register Automata}
We start by formally defining CRAs. We restrict attention to the tropical semiring, so we define CRAs explicitly with the min-plus operations, as opposed to a general semiring as in~\cite{alur2013regular}. Technically, we use a matrix representation, as done in~\cite{Almagor2020WeakCostRegister}.

A \emph{Cost Register Automaton with linear register updates} of dimension $k$ ($k$-CRA) is a tuple
$\cN = \tup{Q,\Sigma,\delta,q_0,F,\upd,\fin}$
with the following components:
\begin{itemize} 
    \item $\tup{Q,\Sigma,\delta,q_0,F}$ is a deterministic finite automaton (DFA), with $\delta:Q\times \Sigma\to Q$ a deterministic transition function.
    \item For each state $q\in Q$ and letter $\sigma\in \Sigma$, the \emph{update}
    $\upd(q,\sigma)$ is a $k\times k$ matrix over $\bbZinf$. We refer to its entries as $\upd(q,\sigma)_{i,j}$ for $1\le i,j\le k$.
    \item For each $q\in F$, the \emph{output registers} are $\fin(q)\subseteq [k]$.
\end{itemize}

We turn to define the semantics of CRAs.
Recall that all matrix products are in the $(\bbZinf,\min,+)$ semiring.  
A \emph{valuation} of the $k$ registers is a row vector $\vec{r}\in \bbZinf^k$.
Given such a valuation and a transition update $\upd(q,\sigma)$, the \emph{next valuation} is $\vec{r'}=\vec{r}\cdot \upd(q,\sigma)$.
Denote $M=\upd(q,\sigma)$ with entries $m_{i,j}$. Unfolding tropical matrix multiplication gives the following register update:
\[
    r_i' = \min\{r_1 + m_{1,i}, r_2 + m_{2,i},\ldots, r_k + m_{k,i}\}
\]

Consider a word $w=\sigma_1\cdots \sigma_n$ and the unique run
\[
    \rho=q_0\xrightarrow{\sigma_1}q_1\xrightarrow{\sigma_2}\cdots\xrightarrow{\sigma_n}q_n
\]
of the DFA part of $\cN$ on $w$.
Then $\rho$ induces a sequence of valuations $\vec{r^0},\ldots,\vec{r^n}$ given by the zero vector $\vec{r^0}=\vec{0}$, and for every $0\le i< n$ we have $\vec{r^{i+1}}=\vec{r^i}\cdot \upd(q_i,\sigma_{i+1})$. 
If $q_n\in F$, the value assigned by $\cN$ to $w$ is $\min\{\vec{r^n}(i)\mid i\in \fin(q_n)\}$.
If $q_n\notin F$ the value is $\infty$.

\begin{remark}[On Initial and Final Valuations]
\label{rmk:init fin valuations}
Usually CRAs are defined with initial and final valuations, whereas we use $\vec{0}$ as the initial valuation, and only allow existing values of registers in the final valuation.
This choice is to keep in line with our choice of not having initial and final weights in WFAs, as discussed in \cref{rmk: initial and final weights}. 

All our results can be easily adapted to the setting where initial and final valuations/weights are added to both models.    

We also remark that it is not crucial to have the accepting states $F$ in the model, since they are overridden by $\fin$. Nonetheless, it is convenient so that the DFA part is complete.
\end{remark}

\subsubsection*{$k$-Width WFAs}
\paragraph*{Configurations}
We start with a standard definition of configurations, omitted from the preliminaries due to space constraints.
A \emph{configuration} of $\cA$ is a vector $\vec{c}\in \bbZinf^Q$ which, intuitively, describes for each $q\in Q$ the weight $\vec{c}(q)$ of a minimal run to $q$ thus far (assuming some partial word has already been read). 
Let $\vec{c_0}$ be the configuration that assigns $0$ to $q_0$ and $\infty$ to $Q\setminus\{q_0\}$.
Intuitively, before reading a word, $\cA$ is in the configuration $\vec{c_0}$. 

We adapt our notations to include a given starting configuration $\vec{c}$, as follows.
Given a configuration $\vec{c}$ and a word $w$, they induce a new configuration $\vec{c'}$ by assigning each state the minimal weight with which it is reachable via $w$ from $\vec{c}$. We denote this by 
$\xconf_{\vec{c}}(w)(q)=\minweight_{\vec{c}}(Q\runsto{w} q)$ for every $q\in Q$.
In particular, $\xconf_{\vec{c_0}}(w)$ is the configuration that $\cA$ reaches by reading $w$ along minimal runs. 
We denote by $\supp(\vec{c})=\{q\mid \vec{c}(q)\neq \infty\}$ the \emph{support} of $\vec{c}$.

Consider a WFA $\cA$. The \emph{width} of $\cA$ is the maximal number of states that are simultaneously reachable in $\cA$. That is, the maximal $k\in \bbN$ such that there is a word $w$ with $|\supp(\xconf_{\vec{c_0}}(w))|=k$. 

As mentioned in \cref{sec:prelim}, in this section, for convenience, we allow a \emph{set} of initial states $Q_0$ in a WFA. Note that this does not affect the width, since this set can always be replaced with a single initial state.
For a set $Q$, we denote by $\binom{Q}{\le k}=\{S\subseteq Q\mid |S|\le k\}$ the set of subsets of $Q$ of size at most $k$. 

\subsection{Equivalence of $k$-CRA and Width-$k$ WFA}
\label{apx:sec:CRA equiv WFA}
A fundamental result in~\cite{alur2013regular} is that CRA are expressively equivalent to WFAs. In one direction, the idea is that a WFA can nondeterministically track each register of a CRA using its states. For the converse, a CRA can simulate a nondeterministic WFA by keeping a register for each state of the WFA, and updating the registers according to the transitions of the WFA. 

We start by refining this result to take into account the number of registers. Specifically, we show that the number of registers corresponds exactly to the width of the WFA. The construction in the first direction (CRA$\to$WFA) is identical to that of~\cite{alur2013regular}, and we only make the observation on the width. In the converse direction, we need to modify the construction somewhat in order to re-use registers, instead of having a register for each state. 
\begin{theorem}
    \label{apx:thm:k CRA equivalent to k width WFA}
    The class of functions representable by a $k$-CRA is exactly that representable by width $k$ WFAs.
\end{theorem}
\begin{proof}
    \noindent \textbf{CRA$\to$WFA}\quad
    Consider a $k$-CRA $\cN=\tup{Q,\Sigma,\delta,q_0,F,\upd,\fin}$.
    We construct a WFA $\cA=\tup{Q',\Sigma,Q_0,\Delta,F'}$
    where the states are $Q' = (Q\times [k])$, the initial states are $Q_0=\{(q_0,i)\mid i\in [k]\}$, and the accepting states are $F' =\{(q,i)\mid q\in F,\ i\in\fin(q)\}$.
    Intuitively, the state $(q,i)$ tracks the state $q$ and register $i$.

    For every transition $q'=\delta(q,\sigma)$ and $M=\upd(q,\sigma)=\{m_{i,j}\}_{1\le i,j\le k}$ we add to $\Delta$ the transitions $((q,i),\sigma,m_{i,j},(q',j))$ for every $j\in [k]$. 
    Intuitively, the update for register $j$ on the transition $q\runsto{\sigma}q'$ is a minimum, part of which is $r_i+m_{i,j}$, we allow the run from copy $i$ to enter copy $j$ with weight of $m_{i,j}$.


    It is now easy to see by induction that $\cA$ tracks the valuations of $\cN$. Specifically, for a word $w=\sigma_1\cdots \sigma_n$ and its corresponding run $\rho:q_0\runsto{w}q_n$ and valuation sequence $\vec{r^1},\ldots,\vec{r^n}$ in $\cN$, we have in $\cA$ that the configuration $\vec{c}=\xconf_{\vec{c_0}}(w)$ satisfies $\vec{c}((q_n,i))=\vec{r^n}_i$, and $\vec{c}((q',i))=\infty$ if $q'\neq q_n$ for all $i\in [k]$. 
    In particular, $\cA$ correctly tracks the valuations, and since the only state-component that has finite value is $q_n$, we also have $|\supp(\xconf_{\vec{c_0}}(w))|=k$, hence $\cA$ has width $k$.

    \smallskip
    \noindent \textbf{WFA$\to$CRA}\quad
    Consider a WFA $\cA=\tup{Q,\Sigma,Q_0,\Delta,F}$ of width $k$. We fix some arbitrary linear order $\prec$ on $Q$. Since $\cA$ has width $k$, then for every word we can write $\supp(\xconf_{\vec{c_0}}(w))=\{q_1,\ldots,q_d\}$ where $q_1\prec \ldots \prec q_d$ for $d\le k$. 

    We construct an equivalent CRA $\cN$. Intuitively, $\cN$ deterministically tracks the subset-construction of $\cA$, noticing that this only involves tracking subsets in $\binom{Q}{\le k}$. Then $\cN$ uses $k$ registers to track the minimal weight to each reachable state. The crux of the construction, and where it differs from that of~\cite{alur2013regular}, is that we do not use a separate register for each state, but rather reuse the same $k$ registers and associate register $i$ with the $i$-th element of the current configuration. We proceed with the formal construction.

    We define $\cN=\tup{S, \Sigma,\delta,s_0,G,\upd,\fin}$ as follows. The states are $S=\binom{Q}{\le k}$. For every $T\in S$ and $\sigma\in \Sigma$ we have $\delta(T,\sigma)=\{q'\in Q \mid \exists q\in T.\ q\runsto{\sigma}q'\}$. The initial state is $s_0=Q_0$. Note that by the width $k$ assumption, $\delta(T,\sigma)\in S$ and $Q_0\in S$.
    The final states are $G=\{T\in S\mid T\cap F\neq \emptyset\}$. 
    The update function is defined as follows. Let $T\in S$ and $\sigma\in\Sigma$ and write $T = \{q_1,\ldots,q_d\}$ 
         and $\delta(T,\sigma) = T' = \{q'_1,\ldots,q'_{d'}\}$ where
    $q_1\prec\cdots\prec q_d$ and $q'_1\prec\cdots\prec q'_{d'}$, and $d,d'\le k$.
    Intuitively, when $\cA$ reaches a configuration whose support is $T$, in $\cN$ register $i$ stores the cost of the $i$-th state $q_i$. After reading $\sigma$, we update so that register $j$ stores the cost of the $j$-th state $q'_j$ in $T'$. This is captured by setting $\upd(T,\sigma)$ to be the matrix $M=\{m_{i,j}\}_{1\le i,j\le k}$ defined by $m_{i,j}=\min\{c \mid (q_i,\sigma,c,q'_j)\in\Delta\}$ (for $1\le i\le d$ and $1\le j\le d'$), and by setting $m_{i,j}=\infty$ for all remaining entries. Note that if there are no transitions $(q_i,\sigma,c,q'_j)\in\Delta$ then the minimum is $\infty$.
    Finally, we define $\fin(T) = \{i\in [k] \mid q_i\in F\}$
    i.e., the set of indices of registers corresponding to final states in $T$.

    It is again simple to see by induction that $\cN$ correctly tracks the configurations of $\cA$, and that the output registers yield the minimal run, thus $\cN$ and $\cA$ are equivalent.
    Clearly $\cN$ has $k$ registers (by definition).
\end{proof}

\section{Undecidability of Width Minimisation in WFA}
\label{apx:sec: width reduction undecidable}
In the remainder of this section we prove our undecidability result:
\widthund*
 Note that by \cref{apx:thm:k CRA equivalent to k width WFA}, we get that CRA register minimisation is also undecidable:
\begin{restatable}{corollary}{corundecidable}
     \label{cor:CRA register minimisation undecidable}
     The following problem is undecidable already for $k=7$: given a $k$-CRA, decide whether there is an equivalent $(k-1)$-CRA.
 \end{restatable}

The starting point of the proof is a reduction given in~\cite{Almagor2020Whatsdecidableweighted} showing the undecidability of the upper-boundedness problem for WFA.
For brevity, we do not describe the full construction, but state its properties in a convenient way for our purpose.
The reduction in~\cite{Almagor2020Whatsdecidableweighted} is from the $0$-halting problem of \emph{two counter machines}. For this paper, it suffices to know that two counter machines are a computational model, and that a certain problem about it, dubbed the $0$-halting problem, is undecidable. 
We can now state the following.
\begin{theorem}[Construction from \cite{Almagor2020Whatsdecidableweighted}]
\label{thm:construction from whats decidable}
Given a two-counter machine $\cM$, we can compute a WFA $\cA=\tup{Q,\Sigma,q_0,\Delta,F}$ with the following
properties:
\begin{enumerate} 
    \item $\cA$ has width $6$.
    \item All the states of $\cA$ are accepting.
    \item There is a letter $@\in \Sigma$ such that $(q,@,0,q_0)$ are the only incoming transitions to $q_0$, and are the only transitions on $@$.
    \item If $\cM$ does not 0-halt, then $\cA(w)\le 0$ for every $w\in \Sigma^*$.
    \item If $\cM$ halts, there is a word $x\in (\Sigma\setminus \{@\})^*$ such that $\minweight_{\cA}(q_0\runsto{x@}q_0)=1$.
\end{enumerate}
In particular, $\cA$ is unbounded from above if and only if $\cM$ 0-halts, in which case $\minweight_{\cA}(q_0\runsto{x@}Q)=\minweight_{\cA}(q_0\runsto{x@}q_0)=1$ (i.e., $q_0\runsto{x@}q_0$ is the only run on $x@$, and has weight 1).
\end{theorem}

Our aim now is to construct a WFA $\cA'$ such that $\cA'$ has width $7$, and there is an equivalent WFA of width $6$ if and only if $\cA'$ is upper bounded.
We obtain $\cA'=\tup{Q',\Sigma',Q'_0,\Delta',F'}$ from $\cA$ as follows (see \cref{fig:reduction 7 width}). The states are $Q'=Q\cup \{q_a,q_{\Cl{1}},\ldots,q_{\Cl{6}}\}$. The alphabet is $\Sigma'=\Sigma\cup \{\sep,\Cl{1},\ldots,\Cl{6},a,\xCl{1},\ldots,\xCl{6}\}$. The initial states are $Q'_0=\{q_0,q_a\}$, and all the states are accepting: $F'=Q'$. The transitions $\Delta'$ are:
\[\begin{split}
\Delta'=\Delta &\cup \{(q_a,\sigma,0,q_a)\mid \sigma\in \Sigma'\setminus\{a\}\}\cup \{(q_0,\sep,0,\Cl{i})\mid i\in [6]\} \\
&\cup \{(q_{\Cl{i}},\Cl{i},-1,q_{\Cl{i}}),(q_{\Cl{i}},\sigma,0,q_{\Cl{i}}) \mid i\in [6],\sigma\notin \{\Cl{i},\xCl{i}\}\}
\end{split}\]
In particular notice that there are no transitions from $q_a$ with $a$, and no transitions from $q_{\Cl{i}}$ with $\xCl{i}$. Thus, we think of $a$ and $\xCl{i}$ as ``killing'' $q_a$ and $q_{\Cl{i}}$, respectively.

Observe that $\cA'$ has width 7. Indeed, since $\cA$ has width 6 (by \cref{thm:construction from whats decidable}), then after reading any prefix that does not contain $\sep$, the reachable states are at most $6$ states from $\cA$, and the state $q_a$. If $\sep$ is read, then the reachable states are at most $\{q_a,q_{\Cl{1}},\ldots, q_{\Cl{6}}\}$.

\begin{figure}[ht]
\centering
\begin{tikzpicture}[
    ->, >=stealth',
    auto,
    semithick,
    node distance=1.8cm,
    every state/.style={circle, draw, minimum size=20pt, inner sep=1pt},
    every loop/.style={looseness=4, min distance=25pt}
  ]

  \node[state,initial above,initial text={}] (qa) at (-3,0.3) {$q_a$};
  \node[state,initial above,initial text={}] (q0) at (0,0.3) {$q_0$};

  \node[state] (c1)   at (2,-0.3) {$q_{\Cl{1}}$};
  \node  (dots) at (4.7,-0.2) {$\ldots$};
  \node[state] (c6)   at (6,-0.3) {$q_{\Cl{6}}$};

  \draw[rounded corners] (-1.6,0.9) rectangle (0.6,-0.9);
  \node at (-1.3,0.6) {$Q$};

  \draw[->] (qa) edge[in=210, out=150, looseness=6]
    node[pos=0.5, left, align=center] {$\sigma,0$} (qa);

  \draw[->, dotted] (q0) edge[in=220, out=270, looseness=7]
    node[pos=0.7,left=-1pt] {$ x@ ,1 $} (q0);

  \draw[->] (q0) to[bend right=10] node[pos=0.6, above] {$\sep,0$} (c1);
  \draw[->] (q0) to[bend right=30] node[pos=0.2, below] {$\sep,0$} (c6);

  \draw[->] (c1) edge[in=-30, out=30, looseness=6]
    node[align=center, right] {$\begin{aligned} 
      \Cl{1},&-1\\[-4pt] 
      \sigma,&0 
    \end{aligned}$} (c1);

  \draw[->] (c6) edge[in=-30, out=30, looseness=6]
    node[align=center, right] {$\begin{aligned} 
      \Cl{6},&-1\\[-4pt] 
      \sigma,&0 
    \end{aligned}$} (c6);

    \node (X1) at ($(c1)+(0,1.3)$) {\Large$\times$};
\draw[->] (c1) -- node[right] {$\xCl{1}$} (X1);

\node (X6) at ($(c6)+(0,1.3)$) {\Large$\times$};
\draw[->] (c6) -- node[right] {$\xCl{6}$} (X6);

\node (Xa) at ($(qa)+(0,-1.1)$) {\Large$\times$};
\draw[->] (qa) -- node[right] {$a$} (Xa);

\end{tikzpicture}
\caption{Construction of the WFA $\cA'$ from $\cA$. The letter $\sigma$ represents all ``non-killing'' letters in the transitions. Killing letters lead to $\times$.
Intuitively, we start either at $q_a$ or at $q_0$ (in $\cA$). In $q_a$ there is a self loop with weight $0$ on everything except $a$. From $q_0$ we can run as $\cA$, but also move to the $\Cl{i}$ components on $\sep$. In state $q_{\Cl{i}}$ we self loop with weight $-1$ on $\Cl{i}$, and with $0$ on all other letters except $\xCl{i}$. If $\cA$ is unbounded, there is a self loop on $q_0$ with the word $x@$.}
\label{fig:reduction 7 width}
\end{figure}

We prove the correctness of the reduction, starting with a brief intuition. First, if $\cA$ is upper bounded, it can be easily seen that $q_a$ is redundant. Thus, we can remove $q_a$ and obtain an equivalent width $6$ WFA. For the converse, we observe that taking a word of the form \[(x@)^n\Cl{1}^{k_1}\Cl{2}^{k_2}\Cl{3}^{k_3}\Cl{4}^{k_4}\Cl{5}^{k_5}\Cl{6}^{k_6}\] can yield runs of very different weights in the seven components $q_a,q_{\Cl{1}},\ldots, q_{\Cl{6}}$. Intuitively, this means that we need at least seven states to track this information. We formalise these arguments in the following.

The formal proof is in two parts. For the easy direction, assume $\cA$ is upper bounded. By \cref{thm:construction from whats decidable} this means that $\cA(w)\le 0$ for every word $w\in \Sigma^*$. We claim that therefore, we can remove the state $q_a$ from $\cA'$ without affecting its function. 
Indeed, for every word $y\in \Sigma'^*$, its run in $q_a$ (if exists, i.e., if $y$ does not contain $a$) has weight $0$. However, its run in the $\cA$ component and the $\Cl{i}$ components has weight at most $0$ (since $a$ has $0$ self loops in these components).
It follows that the run in $q_a$ is never strictly minimal, so $q_a$ can be discarded without changing the function of $\cA'$. Note that this reduces the width of $\cA'$ to $6$, and therefore there exists an equivalent width $6$ WFA.

We turn to the converse direction, namely proving that if $\cA$ is unbounded, then $\cA'$ does not have an equivalent width $6$ WFA. The intuitive initial idea is the following: since $\cA'$ is unbounded, there is a word $\zeta=x@\in \Sigma^*$ such that $q_0\runsto{\zeta}q_0$ with weight $+1$. Then, concatenating it with some $\Cl{i}^*$ can yield a run that increases and then decreases. However, there is also a run on this word in the $q_a$ component that maintains weight $0$. This structure is akin to the well known WFA for $w\mapsto\min\{\numin{\zeta}(w),\numin{a}(w)\}$, which cannot be determinised (see e.g., \cite{chatterjee2010quantitative,almagor2026determinization}). 
This suggests that the $q_a$ component really ``adds width'' to the $\cA$ component. 

The two gadgets that enable us to prove this are the $\Cl{i}$ components and their killing letters. Intuitively, each $\Cl{i}$ ``counts'' (negatively) a different letter, but we can abruptly kill any of them\footnote{We remark that we do not actually need $\xCl{6}$, but we add it to simplify the presentation}. 
It stands to reason that we really need all 6 of them to compute this function.

We henceforth essentially treat $\zeta$ as a single letter, with the behaviour $\minweight(q_0\runsto{\zeta}q_0)=1$. 
Assume by way of contradiction that $\cB=\tup{S,\Sigma',S_0,\Theta,G}$ is a width 6 WFA equivalent to $\cA'$. 
Let $\bigM>12\norm{\cB}$. Intuitively, if two runs are at distance at least $\bigM$ from each other, then after reading at most six letters, the runs cannot yield the same weight, since the upper run can decrease by at most $6\norm{\cB}$, and the lower can increase by at most $6\norm{\cB}$. 
Consider the word 
\[w=\zeta^{6\bigM}\Cl{1}^{5\bigM}\Cl{2}^{4\bigM}\Cl{3}^{3\bigM}\Cl{4}^{2\bigM}\Cl{5}^{\bigM}\] 
and the suffix $x=a\xCl{1}\xCl{2}\xCl{3}\xCl{4}\xCl{5}$. 
Upon reading $w$, the runs of $\cA'$ accumulate $6\bigM$ in each of the $\Cl{i}$ components, then loses weight so that $\minweight_{\cA'}(Q_0\runsto{w}q_{\Cl{i}})=i\bigM$ for every $i\in [6]$. In addition, $\minweight_{\cA'}(Q_0\runsto{w}q_a)=0$.
Then, reading $x$ kills the different components, starting from $q_a$, then $q_{\Cl{1}}$ to $q_{\Cl{6}}$. In particular, this induces ``jumps'' of weight $\bigM$ with each letter, as follows. 
Observe that $x$ has seven prefixes, which we denote $x_0=\epsilon$, $x_1=a$, $x_2=a\xCl{1}$, $x_3=a\xCl{1}\xCl{2},\ldots,x_6=x$. We then have $\cA'(wx_i)=i\bigM$, and therefore also $\cB(wx_i)=i\bigM$.

Recall that $\cB$ has width 6. Let $T=\supp(\xconf_{\vec{c_0}}(w))$ in $\cB$, then $|T|\le 6$. For every $i\in \{0,\ldots, 6\}$ let $t_i\in T$ such that 
\[\minweight_{\cB}(S_0\runsto{wx_i} G)=\minweight_{\cB}(S_0\runsto{w}t_i\runsto{x_i}G)=i\bigM\] That is, $t_i$ is a state such that the minimal run of $\cB$ on $wx_i$ passes through $t_i$ after reading $w$. 
By the pigeonhole principle, there are $i<j$ such that $t_i=t_j$, denoted $t$.
Recall that upon reading a single letter, each run can change its value by at most $\norm{\cB}$. We thus have:
\[\minweight_{\cB}(S_0\runsto{w}t)\le \minweight_{\cB}(S_0\runsto{w}t\runsto{x_i}G)+i\norm{\cB} \le i\bigM +6\norm{\cB}\]
(since $i\le 6$).
Similarly, going via $t=t_j$, we have
\[\minweight_{\cB}(S_0\runsto{w}t)\ge \minweight_{\cB}(S_0\runsto{w}t\runsto{x_j}G)-j\norm{\cB} \ge j\bigM -6\norm{\cB}\]
But then we have $i\bigM +6\norm{\cB}\ge j\bigM -6\norm{\cB}$, so $(j-i)\bigM \le 12\norm{\cB}$, but $j-i\ge 1$, so we get $\bigM\le 12\norm{\cB}$, in contradiction to our choice of $\bigM>12\norm{\cB}$.

We conclude that there is no width $6$ WFA equivalent to $\cA'$.
This completes the correctness of the reduction, and the proof of \cref{thm:width reduction undecidable}.

\end{document}

%% file: macros.tex
\renewcommand{\vec}[1]{\boldsymbol{#1}}

\newcommand{\bs}[1]{\boldsymbol{#1}}

\newcommand{\xconf}{\texttt{xconf}}
\newcommand{\supp}{\texttt{supp}}
\newcommand{\sep}{\$}

\newcommand{\bbN}{\mathbb{N}}
\newcommand{\bbZ}{\mathbb{Z}}
\newcommand{\bbQ}{\mathbb{Q}}

\newcommand{\bbNinf}{\mathbb{N}_{\infty}}
\newcommand{\bbZinf}{\mathbb{Z}_{\infty}}
\newcommand{\cA}{\mathcal{A}}
\newcommand{\cB}{\mathcal{B}}

\newcommand{\cD}{\mathcal{D}}
\newcommand{\cU}{\mathcal{U}}
\newcommand{\cN}{\mathcal{N}}
\newcommand{\cM}{\mathcal{M}}

\newcommand{\Utype}{$\mathfrak{U}$-type\xspace}
\newcommand{\Dtype}{$\mathfrak{D}$-type\xspace}

\newcommand{\tup}[1]{\langle #1 \rangle}

\newcommand{\init}{\texttt{init}}
\newcommand{\fin}{\texttt{fin}}
\newcommand{\pref}{\mathsf{pref}}
\newcommand{\weight}{\mathsf{wt}}
\newcommand{\norm}[1]{\|#1\|}

\newcommand{\runsto}[1]{\xrightarrow{#1}}
\newcommand{\minweight}{\mathsf{mwt}}

\newcommand{\commit}{\textsc{Com}}
\newcommand{\update}{\textsc{Updt}}
\newcommand{\trim}{\nrightarrow}
\newcommand{\keep}{\rightarrow}

\newcommand{\offset}{\mathtt{offset}}

\newcommand{\bigM}{{\bs{\mathfrak{m}}}}

\newcommand{\Cl}[1]{\Circled{#1}}
\newcommand{\xCl}[1]{\cancel{\Cl{#1}}}
\newcommand{\numin}[1]{\#_{#1}}

\newcommand{\upd}{\texttt{upd}}